\documentclass[a4paper,USenglish,cleveref, autoref, thm-restate]{lipics-v2021}

\usepackage{booktabs}  
\usepackage{array}     
\usepackage{longtable} 
\usepackage{pdflscape} 
\usepackage{makecell}
\usepackage{textcomp}
\usepackage[textsize=tiny]{todonotes}
\usepackage{xspace}
\usepackage{bm}
\usepackage{siunitx}
\usepackage{mathtools}

\newlength{\punctuationfootlength}

\newcommand{\treecompalgo}{\textsc{Proton}\xspace}

\graphicspath{{./images/}}

\title{Strongly Hyperbolic Unit Disk Graphs} 

\author{Thomas Bläsius}{Karlsruhe Institute of Technology\\{Karlsruhe, Germany}}{thomas.blaesius@kit.edu}{https://orcid.org/0000-0003-2450-744X}{}

\author{Tobias Friedrich}{Hasso Plattner Institute, University of Potsdam\\{Potsdam, Germany}}{tobias.friedrich@hpi.de}{https://orcid.org/0000-0003-0076-6308}{}

\author{Maximilian Katzmann}{Karlsruhe Institute of Technology\\{Karlsruhe, Germany}}{maximilian.katzmann@kit.edu}{https://orcid.org/0000-0002-9302-5527}{}

\author{Daniel Stephan}{GSV Algorithm Consulting UG (haftungsbeschränkt) \& Co. KG\\{Potsdam, Germany}}{stephan@algorithm-consulting.com}{}{}

\authorrunning{Bläsius et al.} 

\Copyright{Thomas Bläsius, Tobias Friedrich, Maximilian Katzmann, and
  Daniel
  Stephan} 



\category{} 

\funding{This research was funded by the Deutsche
  Forschungsgemeinschaft (DFG, German Research Foundation) -- Grant
  No.
  390859508.}

\nolinenumbers 

\hideLIPIcs  

\EventEditors{}
\EventNoEds{2}
\EventLongTitle{40th International Symposium on Theoretical Aspects of Computer Science}
\EventShortTitle{STACS 2023}
\EventAcronym{STACS}
\EventYear{2023}
\EventDate{March 7--10, 2023}
\EventLocation{Hamburg, Germany}
\EventLogo{}
\SeriesVolume{40}
\ArticleNo{23}

\DeclareMathOperator{\diam}{diam}
\newcommand{\dist}{d}

\DeclareMathOperator{\acos}{acos}
\DeclareMathOperator{\acosh}{acosh}
\DeclareMathOperator{\asinh}{asinh}

\newcommand{\DiskRep}{\phi}
\newcommand{\Additive}{\ell}
\newcommand{\LabelSize}{\lambda}

\begin{document}

\maketitle

\begin{abstract}
  The class of Euclidean unit disk graphs is one of the most
  fundamental and well-studied graph classes with underlying geometry.
  In this paper, we identify this class as a special case in the
  broader class of \emph{hyperbolic unit disk graphs} and introduce
  \emph{strongly hyperbolic unit disk graphs} as a natural counterpart
  to the Euclidean variant.  In contrast to the grid-like structures
  exhibited by Euclidean unit disk graphs, strongly hyperbolic
  networks feature hierarchical structures, which are also observed in
  complex real-world networks.

  We investigate basic properties of strongly hyperbolic unit disk
  graphs, including adjacencies and the formation of cliques, and
  utilize the derived insights to demonstrate that the class is useful
  for the development and analysis of graph algorithms.  Specifically,
  we develop a simple greedy routing scheme and analyze its
  performance on strongly hyperbolic unit disk graphs in order to
  prove that routing can be performed more efficiently on such
  networks than in general.
\end{abstract}

\newpage

\section{Introduction}

Studying networks in terms of \emph{graph classes} based on certain
properties is a fundamental tool in graph theory.  Instead of having
to consider all possible graphs, we can focus on the ones in a certain
class, which allows us to get a more fine-grained understanding of
their structural properties and the complexity of graph problems.
Additionally, it facilitates the development of more efficient
algorithms that are tailored towards the characteristics of the
considered networks.

Different classes can be utilized in different contexts.  For example,
the characteristics of wireless communication networks are captured
naturally in \emph{Euclidean unit disk graphs}~\cite{ccj-u-90,
  hs-b-95}, i.e., graphs where vertices can be identified with disks
of equal size in the Euclidean plane and any two are adjacent if and
only if their disks intersect.  In this paper, we use the following
formalization.  Let $G = (V, E)$ be an undirected graph.  A
\emph{(Euclidean) unit disk representation} of $G$ is a mapping
$\DiskRep \colon V \to \mathbb R^2$ together with a \emph{threshold
  radius} $R$ such that $\{u, v\} \in E$ if and only if the distance
between $\DiskRep(u)$ and $\DiskRep(v)$ is at most~$R$.  Then, the
graph~$G$ is a \emph{(Euclidean) unit disk graph} if it has a unit
disk representation.  In such graphs, the generally NP-complete
problem of finding a maximum clique can be solved in polynomial
time~\cite{ccj-u-90, rs-r-03}, and routing can be performed more
efficiently than in general graphs~\cite{kmrs-rudg-18}.

In this paper, we study a related graph class where the Euclidean
ground space is replaced with the hyperbolic plane.  The result is a
generalization of the Euclidean variant, containing networks with a
broader range of structural properties.  Formally, a graph $G$ is a
\emph{hyperbolic unit disk graph}, if there exists a \emph{hyperbolic
  unit disk representation}
$\DiskRep \colon V \rightarrow \mathbb{H}^2$ together with a threshold
radius $R$, such that $\{u, v\} \in E$, if and only if the hyperbolic
distance between the vertex representations is at most
$\dist_{\mathbb{H}^2}(\DiskRep(u), \DiskRep(v)) \le R$.  We note that
the threshold radius~$R$ is part of the representation and can thus
depend on the graph.  The choice of $R$ does not matter in Euclidean
space, as scaling $R$ and all coordinates $\DiskRep(\cdot)$ by the
same factor yields the same adjacencies.  In contrast, there is no
scaling operation in the hyperbolic plane that leaves relative
distances intact.\footnote{Under the common assumption that the
  curvature is $-1$, such a scaling operation does not exist. The term
  ``unit disk'' is still justified as we could instead fix $R = 1$ and
  allow for different curvatures.}  As a result, the size of the
considered region and the threshold radius \emph{do} have an impact on
the structure of the graphs in the hyperbolic setting.

To understand this effect, which is visualized in
Figure~\ref{fig:hyperbolic-unit-disk-graphs}, first consider some
region, say a disk $D$ of radius $R'$, in the Euclidean plane and
assume we distribute vertices evenly in~$D$.  Then the resulting
Euclidean unit disk graph resembles a \emph{grid-like} structure (with
a density depending on the threshold radius $R$ and the radius $R'$ of
$D$).  That is, in the sparse setting, we only find small cliques,
while separators and treewidth as well as the diameter are large, and
we observe a homogeneity among the vertices, in the sense that all
neighborhoods feature similar characteristics.  Essentially, as in a
grid, the graph looks the same no matter from which vertex it is
viewed.

As the hyperbolic plane resembles the Euclidean plane locally, we can
achieve the same grid-like structures by choosing a very small radius
$R'$ and an even smaller threshold radius~$R$
(Figure~\ref{fig:hyperbolic-unit-disk-graphs}~(left)).  In fact, by
scaling the Euclidean unit disk representation of a graph into a
sufficiently small region, we can realize the same adjacencies in the
hyperbolic plane and obtain the following.

\begin{theorem}\label{thm:euclidean-is-hyperbolic}
  Every Euclidean unit disk graph is a hyperbolic unit disk graph.
\end{theorem}

\begin{figure}[t!]
  \centering
  \includegraphics{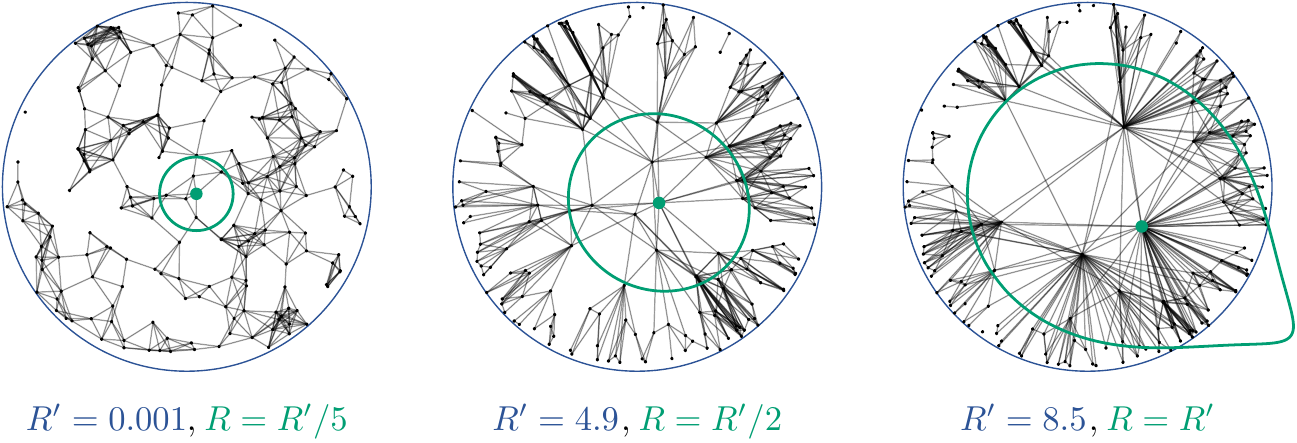}
  \caption{Hyperbolic unit disk graphs with different ground space and
    threshold radii.  The representations have been scaled such that
    the ground spaces \emph{appear} to have the same size, while their
    actual sizes are denoted by $R'$.  \textbf{(Left)} The ground
    space is very small and the threshold radius even smaller, leading
    to grid-like structures. \textbf{(Center)} Ground space and
    threshold radius are increased, hierarchies start form but grid
    like structures remain.  \textbf{(Right)} Ground space and
    threshold have the same large value, leading to hierarchical
    structures.}
  \label{fig:hyperbolic-unit-disk-graphs}
\end{figure}

Beyond that, we can increase the radii $R'$ and $R$.  Then, the
grid-like structures start to vanish and \emph{hierarchical}
structures begin to form
(Figure~\ref{fig:hyperbolic-unit-disk-graphs}~(center)).  Eventually,
we reach the \emph{strongly hyperbolic} setting where only
hierarchical and no grid-like structures remain
(Figure~\ref{fig:hyperbolic-unit-disk-graphs}~(right)).  There,
vertices are rather heterogeneous, with respect to their degree and
what neighborhoods look like.  The diameter is small, while large
cliques can form.  We note that the treewidth is also large, just as
in grid-like graphs.  However, in hierarchical graphs this is an
artifact of the large cliques, while in grid-like graphs we observe
large treewidth \emph{despite} the fact that only small cliques form.
In hierarchical networks vertices connect via hubs, which connect via
larger hubs, and so on.  The hubs explicitly exhibit the hierarchy in
the graph structure.  As a result, the graph looks very different when
viewed from vertices on different levels in the hierarchy.

Formally, we say that a graph is a \emph{strongly hyperbolic unit disk
  graph} if it admits a hyperbolic unit disk representation in which
$\DiskRep$ maps all vertices to points within a disk whose radius
matches the threshold~($R' = R$).  For better understanding, we
recommend using the interactive
visualization\footnote{\url{https://thobl.github.io/hyperbolic-unit-disk-graph/}},
which lets the user change the size of the ground space, allowing to
smoothly transition between Euclidean and strongly hyperbolic unit
disk graphs.

To paint the big picture, hyperbolic unit disk graphs comprise two
extremes: Euclidean unit disk graphs with grid-like structures on one
side and strongly hyperbolic unit disk graphs with hierarchical
structures on the other.  Therefore, if we want to design algorithms
for grid-like structures, it makes sense to analyze them on Euclidean
unit disk graphs.  For hierarchical structures, strongly hyperbolic
unit disk graphs are a good choice.

\paragraph*{Related Concepts.}

To the best of our knowledge, intersection graphs of hyperbolic unit
disks, or hyperbolic unit balls, have so far only been considered by
Kisfaludi-Bak~\cite{k-higqpt-20}.  There, for every $\rho > 0$, a
graph is said to be in the graph class
$\text{UBG}_{\mathbb H^d}(\rho)$ (\emph{UBG = unit ball graph}) if its
vertices can be mapped into $\mathbb H^d$ such that vertices have
distance at most $2\rho$ if and only if they are adjacent.  There are
two core differences compared to our definition of hyperbolic unit
disk graphs.  First, it allows for higher dimensions.  Secondly, it is
parameterized by the radius, i.e., $\text{UBG}_{\mathbb H^d}(\rho)$
describes an infinite family of graph classes rather than a single
class.

This second difference is somewhat subtle but rather important.
Consider the class $\text{UBG}_{\mathbb H^d}(\rho)$ for a fixed radius
$\rho$.  Moreover, assume we want to study graphs in
$\text{UBG}_{\mathbb H^d}(\rho)$ that are sparse; for the sake of
argument, assume constant average degree.  Then, for an increasing
number of vertices $n$, the region of $\mathbb H^d$ spanned by the
vertices has to grow, as otherwise the density of the graph grows with
$n$.  Thus, for sufficiently large $n$, the radius $\rho$ is
arbitrarily small compared to the region spanned by the vertices,
yielding grid-like structures (see discussion above).  Thus, for fixed
$\rho$, large graphs in $\text{UBG}_{\mathbb H^d}(\rho)$ are grid-like
rather than hierarchical.  This means that asymptotic statements for
the classes $\text{UBG}_{\mathbb H^d}(\rho)$ do not translate to the
hierarchical structures in the class of strongly hyperbolic unit disk
graphs.

A second related concept are hyperbolic random graphs~\cite{kpk-h-10},
which are basically random strongly hyperbolic unit disk graphs.  A
hyperbolic random graph is obtained by assigning each vertex a random
point in a disk of radius $R \approx 2 \log(n)$ and connecting two
vertices if and only if their distance is at most $R$.  This yields
graphs that resemble certain real-world networks, as they have small
diameter, high clustering, and a power-law degree distribution whose
exponent can be adjusted using the probability distribution of the
vertex positions.  This is not the case for graphs in
$\text{UBG}_{\mathbb H^d}(\rho)$ as the connection radius cannot
increase with the graph size.  Nevertheless, every hyperbolic random
graph is a strongly hyperbolic unit disk graph and thus any statement
shown for the latter also holds for the former.

Hyperbolic random graphs have also been studied in a noisy setting,
where, with some small probability, distant vertices are adjacent and
close vertices are not adjacent.  Similarly,
Kisfaludi-Bak~\cite{k-higqpt-20} also studies a noisy variant of the
class $\text{UBG}_{\mathbb H^d}(\rho)$.  It would be interesting to
also study (strongly) hyperbolic unit disk graphs in a noisy setting.
This is, however, beyond the scope of this paper and left for future
research.

\paragraph*{Contribution.}

Beyond the generalization of Euclidean unit disk graphs to hyperbolic
unit disk graphs, we identify strongly hyperbolic unit disk graphs as
a natural counterpart to the Euclidean special case and provide the
first insights into their structural and algorithmic properties.  In
particular, we study fundamental criteria relating the coordinates of
vertices to their adjacency and investigate the formation of cliques
(Section~\ref{sec:shdg}).

Using these insights, we follow up on prior empirical efforts towards
understanding how an underlying hyperbolic geometry facilitates
efficient routing on internet-like networks~\cite{bpk-si-10,
  pkbv-gfdsf-10}, and utilize strongly hyperbolic unit disk graphs to
obtain theoretical performance guarantees, proving that routing in
such networks can be performed more efficiently than in general
(Section~\ref{sec:shdg-routing}).  While similar results have been
obtained on the grid-like Euclidean unit disk
graphs~\cite{kmrs-rudg-18}, our analysis covers networks with
hierarchical structures.

In particular, it includes hyperbolic random graphs, which are used to
represent real-world complex networks like the
internet~\cite{bpk-si-10}, where routing plays an important role.  By
developing a simple routing scheme, which is interesting in its own
right, we show that greedy routing on such graphs can be performed
with a stretch of at most $3$, while asymptotically almost surely
requiring at most $\mathcal{O}(\log^4 n)$ bits of storage per vertex
and taking $\mathcal{O}(\log^2 n)$ time per routing decision.

\section{(Strongly) Hyperbolic Unit Disk Graphs}
\label{sec:shdg}

Throughout the paper we consider the \emph{polar-coordinate model} of
the hyperbolic plane $\mathbb{H}^2$.  There, we have a designated
\emph{pole} $O \in \mathbb{H}^2$, together with a \emph{polar axis},
i.e., a reference ray starting at $O$.  A point $p$ is identified by
its \emph{radius} $r(p)$, denoting the hyperbolic distance to~$O$, and
its \emph{angle}~$\varphi(p)$, denoting the angular distance between
the polar axis and the ray from~$O$ through $p$.  In our figures we
interpret these values as polar coordinates in the Euclidean plane.
The disk of radius $R$ centered at $p$ is denoted by $D_R(p)$.  When
$p = O$ we simply write $D_R$.  The hyperbolic distance between points
$p$ and $q$ is given by
\begin{align}
  \label{eq:hyperbolic-distance}
  \dist_{\mathbb{H}^2}(p, q) = \acosh\Big(\cosh\big( r(p) \big)\cosh\big( r(q) \big) - \sinh\big( r(p) \big)\sinh\big( r(q) \big) \cos\big(\Delta_\varphi(p, q)\big)\Big),
\end{align}
where $\cosh(x) = (e^x + e^{-x}) / 2$, $\sinh(x) = (e^x - e^{-x}) /
2$, and $\Delta_\varphi(p, q) = \pi - | \pi - | \varphi(p) -
\varphi(q) ||$ denotes the angular distance between $p$ and $q$.
Without loss of generality, we assume that the representation
$\DiskRep$ of a strongly hyperbolic unit disk graph maps the vertices
into a disk of radius $R$ that is centered at $O$.  For the sake of
readability, we typically associate a vertex $v$ with its mapping
$\DiskRep(v)$ and denote the set of vertices lying in a region
$A \subseteq D_R$ with $V(A)$.

\subsection{Relation to Euclidean Unit Disk Graphs}

We start with the proof of Theorem~\ref{thm:euclidean-is-hyperbolic},
stating that every Euclidean unit disk graph is also a hyperbolic unit
disk graph.  To this end, we utilize the \emph{Poincaré disk} model of
hyperbolic space.  There, the infinite two-dimensional hyperbolic
plane is mapped to the interior of the unit circle in the Euclidean
plane, which is referred to as the Poincaré disk $\mathbb{D}$.  In
this model, points are identified using Cartesian coordinates.  Given
two points $p, q \in \mathbb{D}$, we can compute their hyperbolic
distance by interpreting them as vectors and computing
\begin{align}
  \label{eq:poincare-distance}
  \dist_{\mathbb{D}}(p, q) = 2 \asinh \left( \frac{||p - q||}{\sqrt{(1 - ||p||^2)(1 - ||q||^2)}} \right),
\end{align}
where $|| \cdot ||$ denotes the Euclidean norm.

\begin{proof}[Proof of Theorem~\ref{thm:euclidean-is-hyperbolic}]
  Let $G = (V, E)$ be a Euclidean unit disk graph with representation
  $\DiskRep_E \colon V \rightarrow \mathbb{R}^2$ and threshold radius
  $R_E$.  To prove that $G$ is a hyperbolic unit disk graph, we show
  that there exists a hyperbolic unit disk representation
  $\DiskRep_H \colon V \rightarrow \mathbb{D}$ of $G$ with threshold
  radius $R_H$.  First, note that we can take a disk of radius
  $\rho \in (0, 1)$ in the Euclidean plane and scale the coordinates
  $\DiskRep_E$ and threshold $R_E$ by a positive factor such that all
  vertices lie in this disk, while maintaining a valid Euclidean unit
  disk representation of $G$ with coordinates $\DiskRep_{E}^{\rho}$
  and threshold radius $R_E^{\rho}$.  We now set
  $\DiskRep_H \coloneqq \DiskRep_E^\rho$ for a sufficiently
  small~$\rho$.  In the following, we identify a vertex $v \in V$ with
  its coordinate $\DiskRep_H(v)$, for the sake of readability.
  
  To conclude the proof, it remains to show that there exists an $R_H$
  such that for two vertices $u, v \in V$, we have
  $\dist_{\mathbb{D}}(u, v) \le R_H$ if and only if their Euclidean
  distance is at most $\dist_E(u, v) \le R_E^\rho$.  Note that when
  $u$ and $v$ are not adjacent, there exists a $\tau > 1$ such that
  $\dist_E(u, v) > \tau \cdot R_E^\rho$.  In the following, we
  determine upper and lower bounds on $\dist_{\mathbb{D}}(u, v)$ in
  terms of $\dist_E(u, v)$ and show that, with decreasing $\rho$, they
  approach each other faster than~$\tau$ approaches $1$.  Eventually,
  the bounds are sufficiently tight, such that scaling the lower bound
  by $\tau$ yields something larger than the upper bound, allowing us
  to find a threshold~$R_H$ that fits between the two.

  The hyperbolic distance between $u$ and $v$ can be computed via
  Equation~\eqref{eq:poincare-distance}.  Note that
  $||u - v|| = \dist_E(u, v)$.  Moreover, since $\DiskRep_H$ maps all
  vertices to points in a disk of radius $\rho$, we have
  $0 \le ||u||, ||v|| \le \rho$ and thus
  \begin{align*}
    (1 - \rho^2) = \sqrt{(1 - \rho^2)(1 - \rho^2)} \le \sqrt{(1 - ||u||^2)(1 - ||v||^2)} \le 1,
  \end{align*}
  from which we can derive that
  \begin{align*}
    2 \asinh(\dist_E(u, v)) \le \dist_{\mathbb{D}}(u, v) \le 2 \asinh\left( \frac{1}{1 - \rho^2} \dist_E(u, v)\right).
  \end{align*}
  Since $\asinh(x) \le x$ for all $x \ge 0$, we can set
  $\hat{g}(\rho) \coloneqq 2/(1 - \rho^2)$ and simplify the upper
  bound to
  $\dist_{\mathbb{D}}(u, v) \le \hat{g}(\rho) \cdot \dist_E(u, v)$.
  In order to simplify the lower bound, we use a Taylor-approximation
  of $\asinh(x)$ around $0$ and express the remainder using the
  Lagrange form (see~\cite[Equations 25.2.24 and 25.2.25]{as-hmf-65}),
  in which case there exists a $\xi \in (0, x)$ such that
  \begin{align*}
    \asinh(x) = x - \frac{\xi}{2(1 + \xi^2)^{3/2}} \cdot x^2.
  \end{align*}
  Since the factor is monotonically increasing for $\xi \in [0, 1/2]$,
  we can choose $\rho \le 1/2$ sufficiently small such that
  $0 \le \xi \le \dist_E(u, v) \le 2\rho$, allowing us to bound
  \begin{align*}
    2 \asinh(\dist_E(u, v)) &\ge 2 \cdot \left(\dist_E(u, v) - \frac{\rho}{(1 + 4\rho^2)^{3/2}} \dist_E(u, v)^2 \right) \\
                            &\ge 2\cdot \left(1 - \frac{\rho}{(1 + 4 \rho^2)^{3/2}} \right) \dist_E(u, v) \eqqcolon \check{g}(\rho) \cdot \dist_E(u, v).
  \end{align*}
  where the second inequality follows from the fact that for
  $\rho \le 1/2$ we have $\dist_E(u, v) \le 2\rho \le 1$ and thus
  $\dist_E(u, v)^2 \le \dist_E(u, v)$. We obtain
  \begin{align*}
    \check{g}(\rho) \cdot \dist_E(u, v) \le \dist_{\mathbb{D}}(u, v) \le \hat{g}(\rho) \cdot \dist_E(u, v).
  \end{align*}
  
  Now consider the case where $\{u, v\} \in E$.  Then, we have
  \begin{align*}
    &&\dist_E(u, v) &\le R_E^\rho  \\
    \Leftrightarrow && \hat{g}(\rho) \cdot \dist_E(u, v) &\le \hat{g}(\rho) \cdot R_E^\rho \\
    \Rightarrow && \dist_{\mathbb{D}}(u, v) &\le \hat{g}(\rho) \cdot R_E^\rho,
  \end{align*}
  where the first step follows from the fact that $\hat{g}(\rho) > 0$
  for $\rho < 1$, and the second step is due to the above inequality.
  On the other hand, if $\{u, v\} \notin E$, we have
  \begin{align*}
    && \dist_E(u, v) &> \tau \cdot R_E \\
    \Leftrightarrow && \check{g}(\rho) \cdot \dist_E(u, v) &> \tau \cdot \check{g}(\rho) \cdot R_E^\rho \\
    \Rightarrow && \dist_{\mathbb{D}}(u, v) &> \tau \cdot \check{g}(\rho) \cdot R_E^\rho, 
  \end{align*}
  where the first step is valid since $\check{g}(\rho) > 0$ for
  $\rho > 0$, and the second step, again, follows from the above
  inequality.  It follows that, if there exists a $\rho^* > 0$ such
  that $\hat{g}(\rho^*) < \tau \cdot \check{g}(\rho^*)$, then there
  also exists an
  $R_H \in [\hat{g}(\rho^*), \tau \cdot \check{g}(\rho^*)] \cdot
  R_E^{\rho^*}$, such that $\DiskRep_H$ and $R_H$ yield a valid
  hyperbolic unit disk graph representation of $G$.
  
  Note that $\hat{g}(\rho) < \tau \cdot \check{g}(\rho)$ holds for
  $\rho = 0$.  Since both functions are continuous on $[0, 1)$, so is
  the function $h(\rho) = \tau \cdot \check{g}(\rho) - \hat{g}(\rho)$,
  with $h(0) > 0$.  By applying the $(\varepsilon, \delta)$-definition
  of continuity, we can derive that there exist a $\delta > 0$ such
  that for all $\rho \in (-\delta, \delta)$ we have
  $|h(0) - h(\rho)| \le \varepsilon$ for $\varepsilon = h(0) / 2 > 0$.
  In particular, this means that there exists a $\rho^* > 0$ such that
  $h(\rho^*) > 0$, implying that
  $\hat{g}(\rho^*) < \tau \cdot \check{g}(\rho^*)$.
\end{proof}

\subsection{Adjacency}
\label{sec:shdg-adjacency}

Similar results to the ones described in this subsection have been
determined on hyperbolic random graphs before (see, e.g.,
\cite{gugelmann2012random}).  Here we verify under which requirements
they also hold on strongly hyperbolic unit disk graphs.  By
definition, two vertices in a strongly hyperbolic unit disk graph $G$
are adjacent, if and only if their hyperbolic distance is at most $R$.
Consequently, we can imagine that each vertex~$v$ is equipped with a
neighborhood disk $D_R(v)$.  That is, $N(v) = V(D_R(v))$.  The
following lemma shows that moving such a neighborhood disk closer to
the center of $D_R$ only increases the region of $D_R$ that it covers.

\begin{lemma}
  \label{lem:smaller-radius-increases-disk-cover}
  Let $R$ be a radius and let $p_1, p_2 \in D_R$ be points with
  $r(p_1) \le r(p_2)$ and $\varphi(p_1) = \varphi(p_2)$.  Then,
  $D_R(p_1) \supseteq D_R(p_2) \cap D_R$.
\end{lemma}
\begin{proof}
  \begin{figure}[t]
    \centering
    \includegraphics{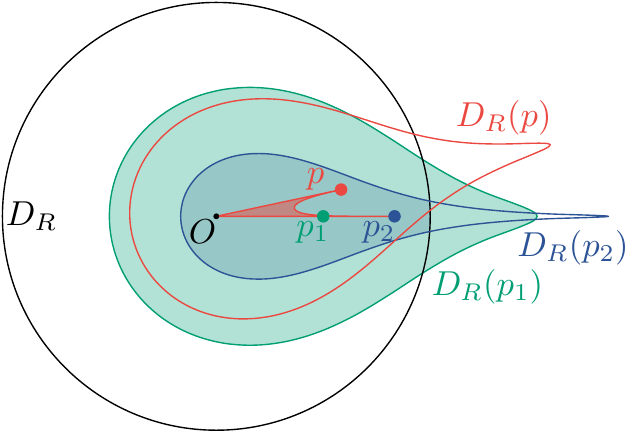}
    \caption{Visualization of the proof
      of~\Cref{lem:smaller-radius-increases-disk-cover}.  Point $p_1$
      has a smaller radius than~$p_2$, both having the same angular
      coordinate.  Consequently, $D_R(p_1)$ (green region) is a
      superset of $D_R(p_2) \cap D_R$ (blue region).  The triangle
      formed by the points $p, p_2,$ and~$O$ is contained in $D_R(p)$
      (both red).}
    \label{fig:smaller-radius-neighborhood}
  \end{figure}
  Let $p \in D_R(p_2) \cap D_R$ be a point and note that
  $\dist_{\mathbb{H}^2}(p, p_2) \le R$.  Now consider the triangle
  spanned by the points $p, p_2$, and the origin $O$.  This triangle
  is completely contained in the disk $D_R(p)$, as
  $\dist_{\mathbb{H}^2}(p, p_2) \le R$ and $r(p) \le R$, as shown
  in~\Cref{fig:smaller-radius-neighborhood}.  Since disks are convex
  and $p_1$ lies on the line from $O$ to $p_2$, it is part of the
  triangle and therefore also contained in the disk.  Consequently,
  $\dist_{\mathbb{H}^2}(p, p_1) \le R$ and thus $p \in D_R(p_1)$.
\end{proof}
Consequently, moving a vertex towards the center does not decrease its
neighborhood.

\begin{corollary}
  \label{lem:smallerRadiusIncreasesNeighborhood}
  Let $G$ be a strongly hyperbolic unit disk graph with radius $R$ and
  let $v_1, v_2$ be vertices with $r(v_1) \le r(v_2) \le R$ and
  $\varphi(v_1) = \varphi(v_2)$. Then, $N(v_1) \supseteq N(v_2)$.
\end{corollary}

In the following, we investigate in greater detail under which
circumstances two vertices are adjacent.  Consider two vertices $v_1$
and $v_2$ in $G$ with radii $r_1$ and $r_2$, respectively.  Clearly,
the two are adjacent if $r_1 + r_2 \le R$.  When $r_1 + r_2 > R$, the
hyperbolic distance between them depends on their angular distance
$\Delta_\varphi(v_1, v_2)$.  More precisely, for vertices of fixed
radii, increasing the angular distance also increases the hyperbolic
distance.  Let $\theta(r_1, r_2)$ denote the angular distance, such
that the hyperbolic distance between $v_1$ and $v_2$ is exactly $R$.
That is, for $\Delta_{\varphi}(v_1, v_2) \le \theta(r_1, r_2)$ we have
$\dist_{\mathbb{H}^2}(v_1, v_2) \le R$, meaning $v_1$ and $v_2$ are
adjacent.  We can compute $\theta(r_1, r_2)$ by using the hyperbolic
distance function in~\Cref{eq:hyperbolic-distance}, setting the
distance equal to $R$, and solving for the angular distance.  That is,
\begin{align}
  \label{eq:theta-angle-definition}
  \theta(r_1, r_2) = \acos\left( \frac{\cosh(r_1)\cosh(r_2) - \cosh(R)}{\sinh(r_1)\sinh(r_2)} \right).
\end{align}

Tight asymptotic bounds on $\theta(r_1, r_2)$ have been derived before
\cite[Lemma 3.2]{k-sahrg-16}.  The following lemma holds for all
$R > 0$.

\begin{lemma}
  \label{lem:tight-max-angle-adjacent}
  Let $R > 0$ and $r_1, r_2 \in (0, R]$ with $r_1 + r_2 \ge R$ be
  given. Then,
  \begin{equation*}
    2\sqrt{e^{R - r_1 - r_2} + e^{-R - r_1 - r_2} - (e^{-2 r_1} + e^{-2 r_2})} \le \theta(r_1, r_2) \le \pi \sqrt{e^{R - r_1 - r_2}}.
  \end{equation*}
\end{lemma}
\begin{proof}
  We start by applying the cosine function on both sides
  of~\Cref{eq:theta-angle-definition}, which makes it easier to deal
  with the right hand side for now.  This yields
  \begin{align}
    \label{eq:solve-distance-for-angle}
    \cos(\theta(r_1, r_2)) = \frac{\cosh(r_1)\cosh(r_2) - \cosh(R)}{\sinh(r_1)\sinh(r_2)}.
  \end{align}
  We consider the upper bound on $\theta(r_1, r_2)$ first.  Note that
  we aim to eventually apply the inverse cosine function to revert the
  above step.  Since this function is monotonically decreasing, we
  first determine a \emph{lower} bound on $\cos(\theta(r_1, r_2))$, in
  order to obtain an upper bound on $\theta(r_1, r_2)$.  Recall that
  $\cosh(x) = (e^{x} + e^{-x})/2$ and $\sinh(x) = (e^{x} - e^{-x})/2$,
  and note that $\sinh(x) \le e^{x} / 2$.  Thus, the above equation
  can be bounded by
  \begin{align*}
    \cos(\theta(r_1, r_2)) &\ge \frac{1/4(e^{r_1} + e^{-r_1})(e^{r_2} + e^{-r_2}) - 1/2(e^{R} + e^{-R})}{1/4e^{r_1 + r_2}} \\
                           &= \frac{e^{r_1 + r_2} + e^{r_1 - r_2} + e^{r_2 - r_1} + e^{-r_1 - r_2} - 2e^{R} - 2e^{-R}}{e^{r_1 + r_2}} \\
                           &= 1 - 2e^{R - r_1 - r_2} + e^{-2r_1} + e^{-2r_2} + e^{-2(r_1 + r_2)} - 2e^{-R - r_1 - r_2}.
  \end{align*}
  We now argue that the remaining expression can be bounded by
  dropping the last four terms since their sum is non-negative.  First
  note that $e^{x} \ge 0$ for all $x \in \mathbb{R}$.  Consequently,
  the second to last term is non-negative and it remains to show that
  $e^{-2r_1} + e^{-2r_2} \ge 2e^{-R - r_1 - r_2}$, which can be done
  by showing that $e^{-2r_1}, e^{-2r_2} \ge e^{-R - r_1 - r_2}$.  In
  the following, we show that this is the case for $e^{-2r_1}$.  The
  proof for $e^{-2r_2}$ is analogous.  Note that $r_1 - r_2 \le R$,
  since $r_1, r_2 \in (0, R]$ by assumption.  It follows that
  $r_1 \le R + r_2$ and thus $e^{-2r_1} \ge e^{-R - r_1 - r_2}$.  We
  can conclude that
  $\cos(\theta(r_1, r_2)) \ge 1 - 2e^{R - r_1 - r_2}$.  The claimed
  upper bound now follows by applying the inverse cosine and observing
  that $\acos(1 - x) \le \pi \sqrt{x / 2}$ holds for all
  $x \in [0, 2]$.

  It remains to prove that the claimed lower bound on
  $\theta(r_1, r_2)$ is valid.  Again, we start
  with~\Cref{eq:solve-distance-for-angle}.  However, this time we
  determine an \emph{upper} bound on $\cos(\theta(r_1, r_2))$.  First,
  we apply the identity
  \begin{align*}
    \cosh(x)\cosh(y) = \sinh(x)\sinh(y) + \cosh(x - y),
  \end{align*}
  which yields
  \begin{align*}
    \cos(\theta(r_1, r_2)) &= \frac{\sinh(r_1)\sinh(r_2) + \cosh(r_1 - r_2) - \cosh(R)}{\sinh(r_1)\sinh(r_2)} \\
                           &= 1 - \frac{\cosh(R) - \cosh(r_1 - r_2)}{\sinh(r_1)\sinh(r_2)}.
  \end{align*}
  Using the definition of $\cosh$ and the fact that
  $\sinh(x) \le e^{x}/2$ to conclude that
  \begin{align*}
    \cos(\theta(r_1, r_2)) &\le 1 - \frac{1/2(e^R + e^{-R}) - 1/2(e^{r_1 - r_2} + e^{r_2 - r_1})}{1/4 e^{r_1 + r_2}} \\
                           &= 1 - 2(e^{R - r_1 - r_2} + e^{-R - r_1 - r_2} - (e^{-2r_2} + e^{-2r_1})) \\
  \end{align*}
  The claim then follows by applying the inverse cosine function and
  observing that $\acos(1 - x) \ge \sqrt{2x}$ is valid for all
  $x \in [0, 2]$.
\end{proof}

We note that, while the above bounds are easier to work with than the
exact function and are generally applicable due to the few constraints
on the considered radii, the lower bound is still a bit tedious to
work with.  However, by introducing some minor requirements, we can
obtain a slightly weaker bound that can be worked with more easily.

\begin{corollary}
  \label{lem:max-angle-adjacent}
  Let $R \ge 1$ and $r_1, r_2 \in (0, R]$ with $r_1 + r_2 \ge R$ and
  $|r_1 - r_2| \le R - 1$ be given. Then,
  \begin{equation*}
    \sqrt{e^{R - r_1 - r_2}} \le \theta(r_1, r_2) \le \pi \sqrt{e^{R - r_1 - r_2}}.
  \end{equation*}
\end{corollary}
\begin{proof}
  The upper bound immediately follows
  from~\Cref{lem:tight-max-angle-adjacent}.  By utilizing the lower
  bound from the same lemma, we obtain
  \begin{align*}
    \theta(r_1, r_2) &\ge 2\sqrt{e^{R - r_1 - r_2} + e^{-R - r_1 - r_2} - (e^{-2 r_1} + e^{-2 r_2})}. \\
                     &\ge 2\sqrt{e^{R - r_1 - r_2} - (e^{-2 r_1} + e^{-2 r_2})}. \\
                     &= 2\sqrt{e^{R - r_1 - r_2} \left(1 - e^{-R}(e^{r_1 - r_2} + e^{- (r_1 - r_2)})\right)},
  \end{align*}
  where the second inequality is valid since
  $e^{-R - r_1 - r_2} \ge 0$.  To prove the claim, it thus suffices to
  show that the remaining negative part is at most $3/4$, which can be
  done as follows.  First note that
  \begin{align*}
    e^{-R}(e^{r_1 - r_2} + e^{-(r_1 - r_2)}) &= 2e^{-R} \cdot \frac{1}{2} (e^{r_1 - r_2} + e^{-(r_1 - r_2)}) = 2e^{-R} \cdot \cosh(r_1 - r_2).
  \end{align*}
  Now note that $\cosh(x)$ is symmetric about the $y$-axis and thus
  $\cosh(r_1 - r_2) = \cosh(|r_1 - r_2|)$.  Moreover, since $\cosh(x)$
  is monotonically increasing for $x \ge 0$, we can utilize the
  assumption that $|r_1 - r_2| \le R - 1$ to conclude
  \begin{align*}
    e^{-R}(e^{r_1 - r_2} + e^{-(r_1 - r_2)}) &\le 2e^{-R} \cdot \cosh(R - 1).
  \end{align*}
  Finally, since $\cosh(x) = 1/2 (e^x + e^{-x}) \le e^x$
  for all $x \ge 0$, we obtain
  \begin{align*}
    e^{-R}(e^{r_1 - r_2} + e^{-(r_1 - r_2)}) &\le 2e^{-R} \cdot e^{R - 1} = 2/e \le 3/4. \qedhere
  \end{align*}
\end{proof}

Apart from the above bounds, we highlight another property of the
function $\theta(r_1, r_2)$, for the special case where $r_1 = r_2$.
\begin{lemma}
  \label{lem:theta-mono-increasing}
  The function $\theta(r, r)$ is monotonically decreasing for
  $r \ge 0$.
\end{lemma}
\begin{proof}
  Consider the definition of $\theta(r_1, r_2)$
  in~\Cref{eq:theta-angle-definition}.  By utilizing the fact that
  $r_1 = r_2 = r$, the equation simplifies to
  \begin{align*}
    \theta(r, r) = \acos\left( \frac{\cosh(r)^2 - \cosh(r)}{\sinh(r)^2} \right).
  \end{align*}
  We can now apply the identities $\cosh(x)^2 = (\cosh(2x) + 1) / 2$
  and $\sinh(x)^2 = (\cosh(2x) - 1) / 2$, both being valid for
  $x \in \mathbb{R}$, to obtain
  \begin{align*}
    \theta(r, r) &= \acos\left( \frac{1/2(\cosh(2r) + 1) - \cosh(r)}{1/2(\cosh(2r) -1)} \right) \\
                 &= \acos\left( \frac{(\cosh(2r) + 1) - 2 \cosh(r)}{\cosh(2r) -1} \right) \\
                 &= \acos\left( \frac{\cosh(2r) - 1 + 2 - 2 \cosh(r)}{\cosh(2r) -1} \right) \\
                 &= \acos\left( 1 - 2 \frac{\cosh(r) - 1}{\cosh(2r) -1} \right).
  \end{align*}
  Further, utilizing the fact that
  \begin{align*}
    \frac{\cosh(x) - 1}{\cosh(2x) - 1} = \frac{1}{2 \cosh(x) + 2},
  \end{align*}
  which is valid for all $x \in \mathbb{R}$, the above term can be
  simplified to
  \begin{align*}
    \theta(r, r) &= \acos\left( 1 - \frac{1}{\cosh(r) + 1} \right).
  \end{align*}
  Note that $\cosh(x)$ is monotonically increasing for $x \ge 0$, and
  so is the argument in the inverted cosine.  The claim follows as
  $\acos$ is monotonically decreasing.
\end{proof}

\subsection{Cliques}
\label{sec:shdg-cliques}

In the following, we examine how the underlying geometry affects the
formation of cliques.  We start by showing that the vertices lying in
$D_R(p)$, having smaller radius than $p$, form two cliques.  More
precisely, we say that a vertex set $S \subseteq V$ can be
\emph{covered by $k$ cliques}, if there exists a partitioning
$S_1, \dots, S_k$ of $S$ such that the induced subgraphs $G[S_i]$ for
$i \in [k]$ are cliques.

\begin{lemma}
  \label{lem:inner-neighborhood-clique}
  Let $G$ be a strongly hyperbolic unit disk graph with radius~$R > 0$
  and let $p \in D_R$ be a point with $r(p) = r$.  Then,
  $V(D_R(p) \cap D_r)$ can be covered by two cliques.
\end{lemma}
\begin{proof}
  Assume without loss of generality that $\varphi(p) = 0$.  We divide
  the region $D_R(p) \cap D_r$ into two halves $A$ and $A'$ containing
  all points with angles in $[0, \pi)$ and $[\pi, 2\pi)$,
  respectively, as illustrated
  in~\Cref{fig:clique-cover-merged}~(left).  The goal now is to show
  that the vertices in~$V(A)$ and the ones in $V(A')$ induce a clique.
  More precisely, we show that this is the case for $A$.  For symmetry
  reasons this then also holds for~$A'$.  Consider two vertices
  $v_1, v_2 \in A$ and assume without loss of generality that
  $\varphi(v_1) \le \varphi(v_2)$.  Since $v_2 \in A \subseteq D_R(p)$
  and since by~\Cref{lem:smaller-radius-increases-disk-cover}
  moving~$D_R(p)$ towards the origin increases the region of $D_R$
  that it covers, we know that $v_2$ is contained in the disk
  $D_R(p')$ for $p' = (r(v_1), 0)$ (dark green
  in~\Cref{fig:clique-cover-merged}~(left)).  It follows that
  $\dist_{\mathbb{H}^2}(p', v_2) \le R$.  Note that $v_1$ has the same
  radius as $p'$ and that
  $\Delta_{\varphi}(p', v_2) \ge \Delta_{\varphi}(v_1, v_2)$.  As
  established above, decreasing the angular distance between two
  points with fixed radii decreases their hyperbolic distance.
  Therefore,
  $\dist_{\mathbb{H}^2}(v_1, v_2) \le \dist_{\mathbb{H}^2}(p', v_2)
  \le R$, meaning $v_1$ and $v_2$ are adjacent.
  \begin{figure}[t]
    \centering
    \includegraphics{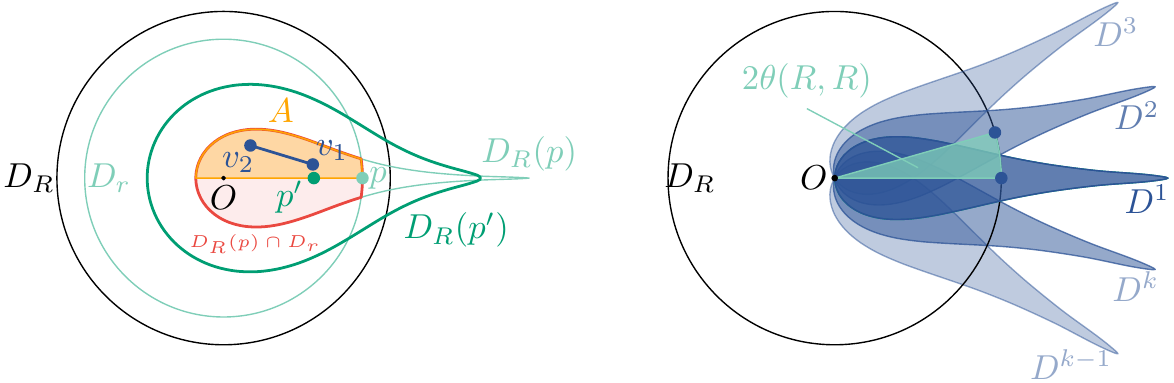}
    \caption{Covering strongly hyperbolic unit disk graphs with
      cliques. \textbf{(Left)} Visualization of the proof
      of~\Cref{lem:inner-neighborhood-clique}.  Vertices $v_1, v_2$
      (blue) are in the half~$A$ (orange) of the region
      $D_R(p) \cap D_r$ (red) and are adjacent. \textbf{(Right)}
      Visualization of the proof
      of~\Cref{lem:strongly-unit-disk-clique-cover}, showing five of
      the $k$ disks $D^1, \dots, D^k$ (blue) and the angular distance
      between two consecutive centers (green).}
    \label{fig:clique-cover-merged}
  \end{figure}
\end{proof}

We note that the above lemma implies that for a vertex $v$, the
neighbors with smaller radius than $v$ form two cliques.  We continue
by investigating the number of cliques required to cover a strongly
hyperbolic unit disk graph.

\begin{lemma}
  \label{lem:strongly-unit-disk-clique-cover}
  Let $G$ be a strongly hyperbolic unit disk graph with radius
  $R > 0$.  Then, $G$ can be covered by
  $\max\{2\pi \sqrt{2}, 2\pi e^{R/2}\}$ cliques.
\end{lemma}
\begin{proof}
  To prove the claim, we utilize the underlying geometry by covering
  the ground space~$D_R$ with a set of $k$ disks $D^1, \dots, D^{k}$,
  such that each $V(D^i)$ for $i \in [k]$ can be covered by two
  cliques.  All of these disks have radius $R$ and their centers lie
  on the boundary of the disk~$D_R$.  The center of the first disk has
  an angular coordinate of $0$.  All other disks $D^i$ are placed at
  an angular distance of~$2\theta(R, R)$ to their predecessor
  $D^{i-1}$ in counterclockwise direction.
  See~\Cref{fig:clique-cover-merged}~(right) for an illustration.  As
  a consequence, the boundaries of two consecutive disks intersect on
  the boundary of $D_R$, which is therefore covered completely by
  the~$k$ disks.  It follows that each vertex is contained in at least
  one disk $D^i$.

  Since by~\Cref{lem:inner-neighborhood-clique} each $V(D^i)$ for
  $i \in [k]$ can be covered by two cliques, it suffices to show that
  $k \le \max\{\pi \sqrt{2}, \pi e^{R/2}\}$ in order to finish the
  proof.  To this end, recall that two consecutive disks are placed at
  an angular distance of~$2\theta(R, R)$.  Consequently, it takes
  $k = 2\pi / (2\theta(R, R)) = \pi / \theta(R, R)$ disks to cover the
  whole disk $D_R$.  Using~\Cref{lem:tight-max-angle-adjacent} we can
  conclude
  \begin{align*}
    \theta(R, R) &\ge 2 \sqrt{e^{-R} + e^{-3R} - 2e^{-2R}} = 2 \sqrt{\left( e^{-R/2} - e^{-3/2 \cdot R} \right)^2} 
    = 2e^{-R/2} (1 - e^{-R} ).
  \end{align*}
  It follows that $k$ can be bounded by
  \begin{align}
    \label{eq:k-bound}
    k = \frac{\pi}{\theta(R, R)} \le \frac{\pi}{2e^{-R/2}\left( 1 - e^{-R} \right)} = \pi e^{R/2} \cdot \frac{1}{2\left(1 - e^{-R} \right)}.
  \end{align}
  We now distinguish between two cases depending on the size of $R$
  and start with~$R < \log(2)$.  Recall that the function
  $\theta(R, R)$ is monotonically decreasing
  in~$R$~(see~\Cref{lem:theta-mono-increasing}).  As a consequence, we
  have $\theta(R, R) \ge \theta(\log(2), \log(2))$.  Then, it follows
  that
  \begin{align*}
    k &\le \frac{\pi}{\theta(\log(2), \log(2))} \le \pi e^{\log(2)/2} \frac{1}{2 \left( 1 - e^{-\log(2)} \right)} = \pi \sqrt{2},
  \end{align*}
  which we account for with the first part of the maximum.  When
  $R \ge \log(2)$, note that we have $(1 - e^{-R}) \ge 1/2$.
  Consequently, we can bound the last fraction in~\Cref{eq:k-bound} by
  $1$, which yields the claim.
\end{proof}

\section{Routing}
\label{sec:shdg-routing}

While finding a path between two vertices in an undirected network is
typically rather simple, the internet is decentralized and does not
allow for the use of a central data structure.  Instead each vertex
can only use local information to perform a \emph{routing decision},
i.e., the decision to which vertex information is forwarded next such
that it eventually reaches the target.  This situation is further
complicated by the fact that the internet consists of billions of
vertices.  In order to be able to handle a network of this scale, a
routing scheme has to be optimized with respect to three criteria: the
\emph{space requirement} (the amount of information that the scheme
uses to forward information), the \emph{query time} (the time it takes
to make a routing decision), and the \emph{stretch} (how much longer
the routed path is compared to a shortest path in the network, where
the \emph{length} of a path denotes the number of contained edges).
Formally, a path between two vertices has \emph{multiplicative
  stretch} $c \ge 1$ if it is at most~$c$ times longer than a shortest
path between them.  An \emph{additive stretch} of $\Additive$ denotes
that the routed path contains at most $\Additive$ more vertices than a
shortest path.  A \emph{multiplicative stretch~$c$ with additive bound
  $\Additive$} denotes that the routed paths have stretch~$c$
\emph{or} additive stretch $\Additive$.  Note that this implies a
multiplicative stretch of $\max \{c, 1 + \Additive\}$.

\subsection{A Brief History of Routing}
\label{apx:routing}

In the following, we summarize the main approaches to adjusting the
trade-off between stretch and space requirements.  The query times of
the considered schemes are at most polylogarithmic.
Figure~\ref{fig:related-work} gives an overview of existing schemes.

\subparagraph{Routing Schemes.}

\begin{figure}[t!]
  \centering
  \includegraphics[width=\linewidth]{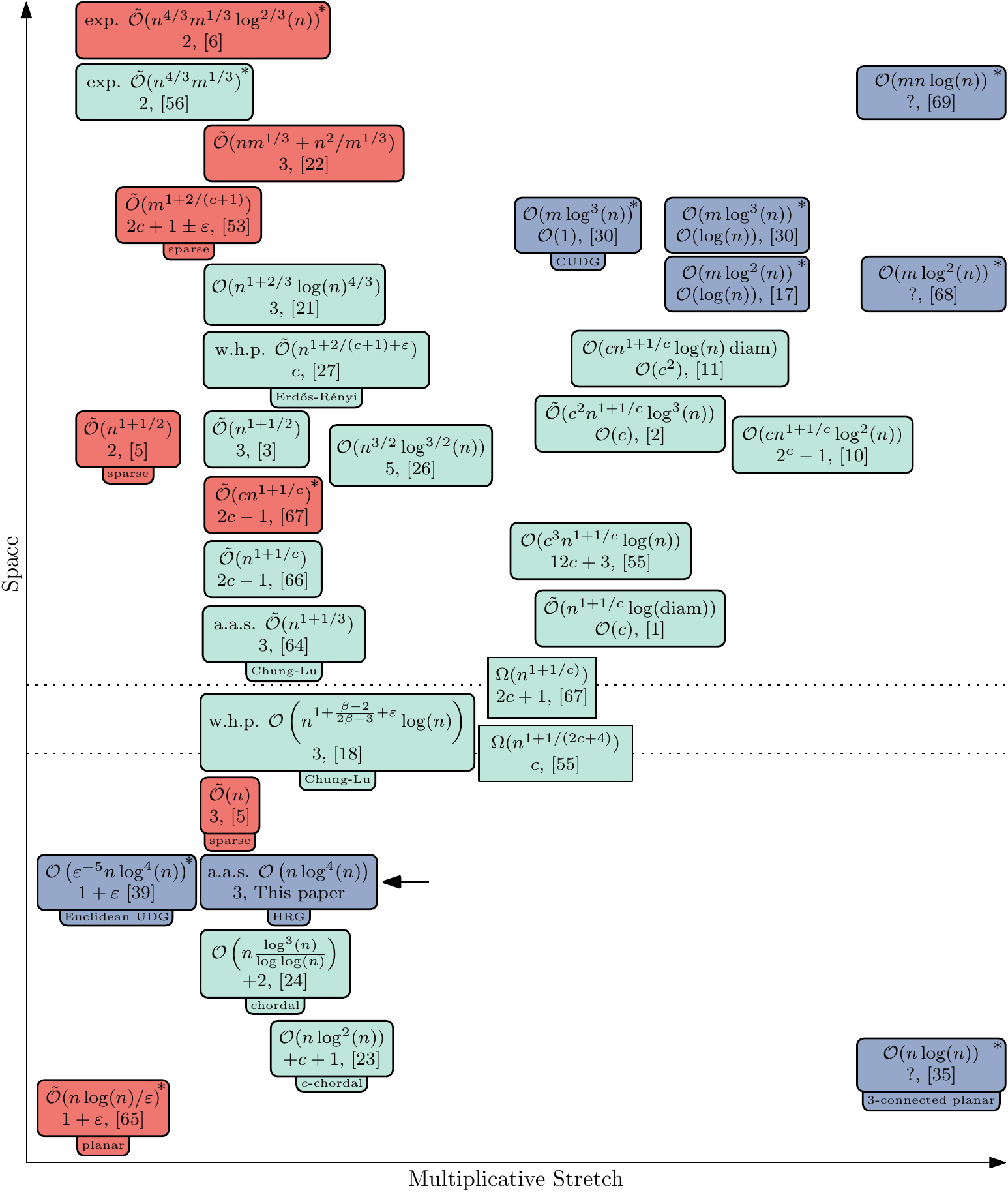}
  \caption{Distance oracles (red), routing schemes (green), and local
    routing schemes (blue) arranged by space requirements (first line)
    and multiplicative stretch (second line).  Additive stretch is
    denoted with a preceding $+$.  An asterisk denotes that the bound
    on the space requirement is adjusted to account for the total
    storage.  Lower bounds are shown with rectangular corners.}
  \label{fig:related-work}
\end{figure}

In general networks, routing with a stretch of $1$, i.e., always
routing along the shortest paths, requires storing
$\Theta(n^2 \log n)$ bits in total~\cite{gp-mrrdn-96}.  The most
commonly used approach to reducing the required space is to only store
shortest path information for certain vertex pairs.  That is, a
representation of a subgraph (typically a tree or a collection of
trees) of the original graph is stored and the routing takes place on
the subgraph.  This is usually done by selecting a set of
\emph{landmark} (or \emph{pivot}) vertices.  Then, for each vertex
only the information about how to get to the closest landmark is
stored~\cite{agm-ricsto-04, agm-ssto-06, agm-cnirms-08, ablp-i-90,
  ap-rpcsto-92, c-crms-01, egp-c-03, pr-dot-14, tz-crs-01}.  These
schemes basically partition the graph based on the landmark vertices.
A related approach starts with a partition and defines the landmarks
afterwards~\cite{pu-tsert-89, rt-c-16}.  The scheme then routes via
the landmarks closest to the source and target.  This general approach
can be optimized in several ways.  First, the network can be
partitioned with several levels of granularity, such that messages
that need to travel larger distances are routed to landmarks whose
associated vertex set is larger~\cite{ablp-i-90, pu-tsert-89}.
Improvements for shorter distances can be obtained by storing the
actual shortest path information for vertices in close vicinity of
each vertex~\cite{ablp-i-90, egp-c-03}.  Moreover, the selection of
the landmarks itself can have an impact on the performance of the
routing scheme.  In general graphs they are typically selected at
random.  A more careful selection can lead to better results on
Erd\"{o}s-R\'{e}nyi random graphs~\cite{ewg-rmsicr-08}, or when
assuming that the network has certain properties like a power-law
degree distribution~\cite{cstw-crsad-12, tyz-crrplg-09}. Similarly,
better results can be obtained on chordal graphs~\cite{d-crsgcg-05,
  dg-icrscg-02, dyl-ctsg-04}.

Closely related to routing schemes are \emph{approximate distance
  oracles}.  There, we are only interested in the length of a short
path instead of the path itself.  Again, the commonly used techniques
are based on landmarks~\cite{ag-b-13, dhz-a-96, prt-nidosg-12,
  tz-ado-05}, and compared to general networks, better results can be
obtained when assuming that the considered graphs have certain
properties like being planar~\cite{t-coradpd-04}, or being sparse
(although at the expense of an increased query time)~\cite{agh-a-11}.

In general, routing schemes and distance oracles are based on one
central data structure that holds the information required when
forwarding messages, which can become an issue with increasing network
size as it has been shown that, on general graphs, achieving a stretch
of $c \ge 1$ requires a data structure of size
$\Omega(n^{1 + 1/(2c + 4)})$~\cite{pu-tsert-89}.  One approach to
overcoming this problem is to consider local routing schemes instead.

\subparagraph{Local and Greedy Routing Schemes.}

In local routing schemes the routing information is distributed and
each vertex can only use its own information to forward messages.  One
approach to achieving this are \emph{interval routing schemes}, where
each vertex is equipped with a mapping from its outgoing edges to a
partition of the vertices in the graph and the message is forwarded
along the edge whose assigned vertex set contains the
target~\cite{egp-c-03, pu-tsert-89, sk-lirn-85}.

Another popular approach is \emph{geographic} or \emph{greedy}
routing.  There, each vertex is assigned a coordinate in a metric
space and a message is routed to a neighbor that is closer to the
target with respect to the metric.  While initially being motivated by
real-world networks with actual geographic locations~\cite{kk-g-00,
  tk-otrrd-84}, later adaptations assigned \emph{virtual
  coordinates}~\cite{rrp-grli-03}.

In addition to the previously mentioned criteria, greedy routing is
also evaluated regarding the success rate, since the virtual
coordinates may be assigned such that forwarding messages greedily
leads to a dead end.  Even simple graphs like a star with six leaves
cannot be embedded in the Euclidean plane such that greedy routing
always succeeds~\cite{pr-crgr-05}.  Worse yet, even if a graph admits
a greedy embedding into the Euclidean plane, there are graphs that
require $\Omega (n)$ bits per coordinate~\cite{adf-s-12}.  However, it
was shown that delivery can be guaranteed on every graph when
embedding it in the hyperbolic plane~\cite{k-gruhs-07}.
Unfortunately, due to the properties of hyperbolic space, this
requires high precision coordinates, which leads to an increased space
requirement~\cite{bfkk-henogr-20}.  While attempts have been made to
reduce the coordinate size~\cite{eg-s-08}, it has been shown that this
remains an open problem~\cite{bfkk-henogr-20, kk-stgn-15}.  However, in many
greedy routing schemes the space per coordinate is at most
poly-logarithmic \cite{bc-crplgas-06, cpfv-grwwms-14, eg-sggruhg-11,
  gs-sggrep-09, m-dgraahn-07, wp-srvge-09, zg-gregss-11}.  

Unfortunately, not much is known about stretch in local routing
schemes.  On Euclidean unit disk graphs there is a greedy approach
that obtains constant stretch~\cite{fpw-grbs-09} and there exists an
improvement to a stretch of $1+\varepsilon$ while storing
$\mathcal{O}(\varepsilon^{-5} \log^4(n))$ bits at each vertex, which
relies on message headers that can be adapated during the routing
process~\cite{kmrs-rudg-18}.  On graphs of bounded hyperbolicity, one
can obtain an additive stretch of
$\mathcal{O}(\log n)$~\cite{gl-dlhg-05} and on general graphs a
logarithmic bound on the multiplicative stretch is
known~\cite{cpfv-grwwms-14, fpw-grbs-09}.  However, it has been
observed that greedy routing schemes can achieve much better stretch
in practice, which we discuss in the following.

\subparagraph{Routing in Practice.}

Real-world networks rarely resemble the worst cases considered in the
above mentioned results.  More realistic insights can be obtained by
analyzing networks whose properties resemble those of real-world
graphs, like the small-world phenomenon~\cite{k-swp-00}.  In
particular, better trade-offs between stretch and space have been
obtained on sparse graphs~\cite{agh-a-11, prt-nidosg-12}, and Chung-Lu
random graphs~\cite{acl-rgmplg-01}.  There, the best known space bound
of $\mathcal{\tilde{O}}(n^{1 + 1/2})$ for a stretch of $3$ on general
graphs~\cite{tz-crs-01}, was improved to
$\mathcal{O}(n^{1 + (\beta - 2)/(2\beta - 3) + \varepsilon})$, with
high probability, for power-law exponent $\beta \in (2, 3)$ and
$\varepsilon > 0$~\cite{cstw-crsad-12}.  Experiments on internet-like
networks further indicate that the landmark-based routing schemes due
to Thorup and Zwick~\cite{tz-crs-01} yield a rather low stretch of
about $1.1$ while the information stored at the vertices is small as
well~\cite{kfy-ci-04, kcfb-cri-07}.  Similar results have been
obtained in experiments on internet topologies and random graphs with
power-law degree distributions~\cite{agh-a-11, cstw-crsad-12,
  tyz-crrplg-09}.

Additionally, it was observed that greedy routing works remarkably
well on internet graphs, when assuming an underlying hyperbolic
geometry.  There, a network is embedded into the hyperbolic plane and
a message is always forwarded to the neighbor with the smallest
hyperbolic distance to the target.  Then, delivery is not guaranteed
but experiments show that this achieves success rates of at least
$97\%$ and a stretch of about $1.1$ on internet
topologies~\cite{bpk-si-10, pkbv-gfdsf-10}.  Partly motivated by this
Krioukov et al. introduced the hyperbolic random graph (HRG) model to
represent real-world networks like the internet~\cite{kpk-h-10}.  For
a generalized version it was shown that a greedy routing can succeed
with constant probability, while achieving an average stretch of $1 +
o(1)$, almost surely~\cite{bkl-graswp-17}.  Nevertheless it remains an
open question whether, in a addition to the small stretch, greedy
routing on realistic representations of internet-like graphs can be
implemented, such that delivery is always guaranteed, while keeping
the space requirement low.  In this paper, we answer this question by
developing a greedy routing scheme that always succeeds with small
stretch.  Additionally, the space requirement is small on networks
with an underlying hyperbolic geometry.

In the following, we combine standard techniques in routing to obtain
a simple routing scheme and analyze its performance on strongly
hyperbolic unit disk graphs.  For the special case of hyperbolic
random graphs, a graph model that is used to represent complex
networks like the internet~\cite{kpk-h-10}, it improves below the
above mentioned performance lower bounds.

\subsection{A Greedy Routing Scheme}

A standard approach to routing in a decentralized setting is
\emph{greedy routing}, where the idea is to always forward a message
to a neighbor of a vertex that is closer to the target.\footnote{The
  term \emph{greedy} often refers to routing to a neighbor
  \emph{closest} to the target.  For us \emph{closer} is sufficient.}
When designing a greedy routing scheme, we therefore need to compute
distances between vertices and select a suitable neighbor with respect
to these distances.  For a graph $G = (V, E)$ let
$\dist \colon V \times V \rightarrow \mathbb{R}_{\ge 0}$ be a
semi-metric on $G$.  That is, for all $s, t \in V$ we have
$\dist(s, t) \ge 0$, $\dist(s, t) = 0$ if and only if $s = t$, and
$\dist(s, t) = \dist(t, s)$.  We say that a greedy routing scheme
routes \emph{with respect to} $\dist$, if at $s$ a message to $t$ is
forwarded to a neighbor $v$ of $s$ where $\dist(v, t) < \dist(s, t)$.
Depending on $\dist$ such a neighbor may not exist and the message
cannot be forwarded, which is called \emph{starvation}.  In contrast,
a routing scheme with guaranteed delivery is called
\emph{starvation-free} and it is known that greedy routing is
starvation-free, if at every vertex $s \neq t$ there is a neighbor $v$
with $\dist(v, t) < \dist(s, t)$ (see e.g.~\cite{zg-g-13}).  Moreover,
we say that~$\dist$ is \emph{integral} if it maps to the natural
numbers, i.e., $\dist \colon V \times V \rightarrow \mathbb{N}$.  Note
that, if $\dist$ is integral and routing with respect to $\dist$ is
starvation-free, the distance to the target decreases by at least one
in each step.  Thus, the length of the routed path between $s$ and $t$
is bounded by $\dist(s, t)$.  When this is the case, we say that
routing with respect to $\dist$ is \emph{$\dist$-bounded}.

Given a connected graph $G$, a natural choice for determining a
distance between $s$ and~$t$ is to use the length of a shortest path
between them, which we denote with $\dist_G(s, t)$.  Routing with
respect to $\dist_G$ is starvation-free and yields perfect stretch.
However, $\dist_G$ cannot be computed while simultaneously keeping the
required space and query time low~\cite{gp-mrrdn-96}.  Therefore, we
relax the constraint on routing with respect to exact graph distances
and use upper bounds instead.  This can be achieved by taking a
subgraph $G'$ of $G$ and routing on $G$ with respect to~$\dist_{G'}$.
The stretch of the resulting routing scheme depends on how well the
distances in~$G'$ approximate the distances in~$G$.  Unfortunately,
finding a subgraph with good stretch is hard in general~\cite{c-n-94,
  ps-g-89}.  However, instead of routing with respect to the distances
in a single subgraph, we can combine the distances in multiple
subgraphs.  To obtain a good stretch, it then suffices to find
low-stretch subgraphs for small parts of the graph.  To formalize
this, we use a $(c, \Additive, k)$-\emph{graph-cover} $\mathcal{C}$ of
$G$, which is a collection of subgraphs of $G$, such that for all
$s, t \in V$ there exists a connected subgraph~$G'$ in $\mathcal{C}$
with $\dist_{G'}(s, t) \le c \cdot \dist_G(s, t)$ or
$\dist_{G'}(s, t) \le \dist_G(s, t) + \Additive$, and every vertex
$v \in V$ is contained in at most $k$ graphs in $\mathcal{C}$.  We say
that~$\mathcal{C}$ has multiplicative stretch $c$ with additive bound
$\Additive$.  For two vertices $s$ and $t$ we define
$\dist_{\mathcal{C}}(s, t) = \min_{G' \in \mathcal{C}} \dist_{G'}(s,
t)$.

\begin{lemma}
  \label{lem:graph-cover-dist-stretch}
  Let $G$ be a graph and let $\mathcal{C}$ be a
  $(c, \Additive, k)$-graph-cover of $G$.  Then, greedy routing on $G$
  with respect to $\dist_{\mathcal{C}}$ has multiplicative stretch $c$
  with additive bound $\Additive$.
\end{lemma}
\begin{proof}
  Let $s \neq t$ be two vertices.  To prove the claim, we need to show
  that an $s$-$t$-path obtained by greedily routing with respect to
  $\dist_{\mathcal{C}}$ has length at most $c \cdot \dist_G(s, t)$ or
  $\dist_G(s, t) + \Additive$.  To this end, we prove that the
  resulting routing scheme is $\dist_{\mathcal{C}}$-bounded.  The
  claim then follows, due to the fact that consequently the routed
  $s$-$t$-path has length at most
  $\dist_{\mathcal{C}}(s, t) = \min_{G' \in G} \dist_{G'}(s, t)$ and
  the fact that there exists a $G' \in G$ with
  $\dist_{G'}(s, t) \le c \cdot \dist_{G}(s, t)$ or
  $\dist_{G'}(s, t) \le \dist_{G}(s, t) + \Additive$ by assumption.

  Since $\dist_{\mathcal{C}}$ is the minimum of integral semi-metrics,
  it is itself an integral semi-metric.  Therefore, it suffices to
  show that routing with respect to $\dist_{\mathcal{C}}$ is
  starvation-free, which is the case, if for every two vertices $s
  \neq t$ there exists a neighbor $v$ of $s$ in $G$ with
  $\dist_{\mathcal{C}}(v, t) < \dist_{\mathcal{C}}(s, t)$.  Consider
  the connected subgraph $G' \in \mathcal{C}$ for which $\dist_{G'}(s,
  t) = \dist_{\mathcal{C}}(s, t)$.  Then, there exists a shortest path
  from $s$ to $t$ in $G'$.  For the successor $v$ of $s$ on this path,
  it holds that $\dist_{G'}(v, t) = \dist_{G'}(s, t) - 1$ and thus
  $\dist_{\mathcal{C}}(v, t) \le \dist_{G'}(s, t) - 1 =
  \dist_{\mathcal{C}}(s, t) - 1 < \dist_{\mathcal{C}}(s, t)$.
  Finally, since $G'$ is a subgraph of $G$, it follows that $v$ is
  also a neighbor of $s$ in $G$.
\end{proof}

To show that $\dist_{\mathcal{C}}$ can be computed efficiently, we use
\emph{distance labeling schemes}~\cite{gppr-d-04}.  Such a scheme
implements a semi-metric $\dist$ by assigning each vertex a
\emph{distance label}, such that for two vertices $s, t$ we can
compute $\dist(s, t)$ by looking at their distance labels only.  The
\emph{label size} of a distance labeling scheme denotes the maximum
number of bits required to represent the label of a vertex.  The
\emph{query time} denotes the time it takes to compute $\dist$ using
the labels.  Given a graph-cover $\mathcal{C}$, a distance labeling
scheme that implements $\dist_{\mathcal{C}}$ can be obtained by
combining distance labeling schemes for the contained subgraphs.

\begin{lemma}
  \label{lem:graph-cover-dist-label-scheme}
  Let $G$ be a graph and let $\mathcal{C}$ be a
  $(c, \Additive, k)$-graph-cover of $G$ such that for every
  $G' \in \mathcal{C}$ there exists a distance labeling scheme that
  implements $\dist_{G'}$ with label size $\LabelSize$ and query time
  $q$.  Then, there exists a distance labeling scheme for $G$ that
  implements $\dist_{\mathcal{C}}$ with label size
  $\mathcal{O}(k (\LabelSize + \log k + \log n))$ and query time
  $\mathcal{O}(k q)$.
\end{lemma}
\begin{proof}
  We assign each subgraph $G' \in \mathcal{C}$ a unique
  \emph{graph-ID} in $[|\mathcal{C}|]$ and compute the distance labels
  for all vertices in $G'$.  By combining the distance labels with the
  corresponding graph-ID, we obtain an \emph{identifiable distance
    label} that can be used to uniquely identify to which graph a
  distance label belongs.  The label of a vertex $v$ is then obtained
  by collecting the identifiable distance labels of $v$ for all
  subgraphs that $v$ is contained in and sorting them by graph-ID.

  The label size can now be bounded as follows.  Since each vertex $v$
  is contained in at most $k$ subgraphs, we can conclude that
  $|\mathcal{C}| \le kn$.  Therefore, the graph-IDs can be encoded
  using $\mathcal{O}(\log k + \log n)$ bits.  Moreover, by assumption
  the distance labels in the subgraphs can be represented using
  $\LabelSize$ bits.  It follows that a single identifiable distance
  label takes $\mathcal{O}(\LabelSize + \log k + \log n)$ bits.
  Again, since every vertex is contained in at most $k$ subgraphs,
  $v$'s label consists of at most $k$ identifiable distance labels.
  Consequently, the label size is bounded by
  $\mathcal{O}(k (\LabelSize + \log k + \log n))$.

  It remains to bound the query time.  Given the collection of
  identifiable distance labels of two vertices, we can identify the
  ones with matching graph-IDs in time $\mathcal{O}(k)$, since they
  are sorted by graph-ID.  For each match we compute the distance in
  the corresponding subgraph in time $q$.  Afterwards the minimum
  distance can be found in $\mathcal{O}(k)$ time.  It follows that
  $\dist_{\mathcal{C}}$ can be computed in time $\mathcal{O}(k q)$.
\end{proof}

In order to perform a routing decision efficiently, we want to avoid
performing a linear search over all neighbors.  To this end, we need
to be able to identify a neighbor directly, which can be done by
assigning each neighbor $v$ of~$s$ a unique \emph{port}
$p_s(v) \colon N(s) \rightarrow \{1, \dots, n\}$.  Finding a neighbor
of $s$ that is closer to a target $t$ with respect to a semi-metric
$\dist$ then boils down to determining the corresponding port.  To
this end, we can use a \emph{port labeling scheme} that implements
$\dist$.  Such a scheme assigns each vertex in a graph a \emph{port
  label} such that we can determine the port of a neighbor of $s$ that
is closer to $t$ with respect to $\dist$, by only looking at the port
labels of $s$ and $t$.  The corresponding label sizes and query times
are defined analogous to how they are defined for distance labels.

Given a graph-cover $\mathcal{C}$, we can combine distance and port
labels of the subgraphs in the cover, to obtain a port labeling scheme
that implements $\dist_{\mathcal{C}}$.

\begin{lemma}
  \label{lem:graph-cover-port-label-scheme}
  Let $G$ be a graph and let $\mathcal{C}$ be a
  $(c, \Additive, k)$-graph-cover of $G$ such that for every
  $G' \in \mathcal{C}$ there exist distance and port labeling schemes
  that implement $\dist_{G'}$ with label size $\LabelSize$ and query
  time $q$.  Then, there exists a port labeling scheme for $G$ that
  implements $\dist_{\mathcal{C}}$ with label size
  $\mathcal{O}(k (\LabelSize + \log k + \log n))$ and query time
  $\mathcal{O}(k q)$.
\end{lemma}
\begin{proof}
  For every vertex $s$ in $G$ we fix a port assignment for the
  neighbors of $s$.  Afterwards, we assign the same ports in the
  subgraphs $G'$ of $G$ in $\mathcal{C}$ that $s$ is contained in.
  More precisely, if $v$ is a neighbor of $s$ in $G'$, then the port
  $p_s(v)$ is identical in $G$ and $G'$.  As a consequence, we can use
  a port labeling scheme in $G'$ to determine the port of a neighbor
  of $s$ in $G$.

  Now consider the distance labeling scheme described in
  Lemma~\ref{lem:graph-cover-dist-label-scheme}, where we assigned
  each subgraph $G' \in \mathcal{C}$ a unique graph-ID and computed
  distance labels for all vertices in all subgraphs to obtain
  identifiable distance labels.  In addition, we now compute port
  labels for all vertices in all subgraphs.  By combining them with
  the previously obtained identifiable distance labels, we obtain
  \emph{identifiable distance port labels}.  As before, the label of a
  vertex $v$ then consists of the collection of identifiable distance
  port labels of $v$ in all subgraphs that $v$ is contained in, sorted
  by graph-ID.

  We continue the proof by showing that, given the labels of two
  vertices $s \neq t$, we can compute the port of a neighbor of $s$ in
  $G$ that is closer to $t$ with respect to $\dist_{\mathcal{C}}$.  As
  described in the proof of
  Lemma~\ref{lem:graph-cover-dist-label-scheme}, we can use the labels
  to find the graph-ID of a subgraph $G'$ of $G$ for which
  $\dist_{G'}(s, t) = \dist_{\mathcal{C}}(s, t)$.  We then use the
  corresponding port labels of $s$ and $t$ to determine the port
  $p_s(v)$ of a neighbor $v$ of $s$ that is closer to $t$ with respect
  to $\dist_{G'}$.  Clearly, we have $\dist_{\mathcal{C}}(v, t) \le
  \dist_{G'}(v, t) < \dist_{G'}(s, t) = \dist_{\mathcal{C}}(s, t)$.
  Moreover, since $p_s(v)$ is identical in $G'$ and $G$, it follows
  that $p_s(v)$ is a suitable port in $G$.

  The label size can be bounded as follows.  By
  Lemma~\ref{lem:graph-cover-dist-label-scheme} we can encode all $k$
  identifiable distance labels stored at a vertex using
  $\mathcal{O}(k(\LabelSize + \log k + \log n))$ bits.  Since a single
  identifiable distance label is extended with a port label that takes
  at most $\LabelSize$ bits, it follows that the label size increases
  by an additive $\mathcal{O}(k \LabelSize)$, yielding a size of
  $\mathcal{O}(k(\LabelSize + \log k + \log n))$ bits.

  It remains to bound the query time.  Again, as described in the
  proof of Lemma~\ref{lem:graph-cover-dist-label-scheme}, determining
  the graph-ID of the subgraph $G'$ for which $\dist_{G'}(s, t) =
  \dist_{\mathcal{C}}(s, t)$ takes $\mathcal{O}(k q)$ time.  Computing
  the port $p_s(v)$ of a suitable neighbor $v$ of $s$ then takes an
  additional time $q$.  Consequently, the query time is $\mathcal{O}(k
  q)$.
\end{proof}

We are now ready to combine the above results to obtain our greedy
routing scheme.  To this end, we need to find
$(c, \Additive, k)$-graph-covers with small values for $c$,
$\Additive$, and $k$, as well as distance and port labeling schemes
with small label sizes and query times, as all of these properties
affect the performance of the routing scheme.  While distance labeling
schemes require large labels in general graphs~\cite{gppr-d-04},
better results can be obtained by restricting the graph-cover to only
contain trees as subgraphs.  Such a cover is then called
\emph{tree-cover}.  Tree-covers are standard in
routing~\cite{akp-besfdp-94, ap-rpcsto-92, egp-c-03, fpw-grbs-09,
  hst-fgfr-14, tczy-tcbgrgd-10, tz-crs-01}, and while it is known that
greedy routing with respect to $\dist_{\mathcal{C}}$ for a
$(c, \Additive, k)$-tree-cover $\mathcal{C}$ is starvation-free (see
e.g.,~\cite{hst-fgfr-14}), we also know that the resulting routing
scheme has stretch $c$ with additive bound $\Additive$ due to
Lemma~\ref{lem:graph-cover-dist-stretch}.  Moreover, for trees there
are distance and port labeling schemes with $\mathcal{O}(\log^2 n)$
bit labels and constant query time~\cite{fgnw-odlst-17, tz-crs-01}.
Together with Lemma~\ref{lem:graph-cover-port-label-scheme} we obtain
the following theorem.

\begin{theorem}
  \label{thm:cover-based-routing}
  Given a $(c, \Additive, k)$-tree-cover of a graph $G$, greedy
  routing on $G$ can be implemented such that the routing scheme is
  starvation-free, has stretch $c$ with additive bound $\Additive$,
  stores $\mathcal{O}(k(\log^2 n + \log k))$ bits at each vertex, and
  takes $\mathcal{O}(k)$ time for a routing decision.
\end{theorem}

To complete our scheme, we propose an algorithm that computes a
tree-cover with bounded stretch.  It is an adaptation of an algorithm
previously used to compute graph spanners~\cite{c-facsps-98}.  We
start with the following lemma, describing a situation that is
exploited by our algorithm.

\begin{lemma}
  \label{lem:u-v-covering-ball-tree}
  Let $G = (V, E)$ be a graph, $u, v \in V$, and let $H$ be an induced
  subgraph that contains all vertices on a shortest $uv$-path $P$ in
  $G$.  Let $T$ be a partial shortest-path tree in $H$ rooted at $t$
  that contains $u$ and $v$.  Then, for every vertex $w$ in $T$ that
  lies on $P$, $\dist_T(u, v) \le \dist_G(u, v) + 2\dist_H(t, w)$.
\end{lemma}
\begin{proof}
  Let $t'$ be the lowest common ancestor of $u$ and $v$ and consider
  the paths $P_u$ and $P_v$ from $t'$ to $u$ and $v$, respectively.
  Note that $P_u$ and $P_v$ are shortest paths in $H$ as they are
  descending paths in a shortest-path tree.  Thus, $\dist_T(u, v) =
  |P_u| + |P_v| = \dist_H(t', u) + \dist_H(t', v)$.

  Observe that clearly $\dist_H(t', u) \le \dist_H(t, u)$.  Moreover,
  by the triangle inequality, we have $\dist_H(t, u) \le \dist_H(t, w)
  + \dist_H(w, u)$.  Analogously for $v$, we obtain $\dist_H(t', v)
  \le \dist_H(t, w) + \dist_H(w, v)$.  Thus, we get
  \begin{align*}
    \dist_T(u, v)
    &= \dist_H(t', u) + \dist_H(t', v) \le \dist_H(t, w) + \dist_H(w, u) + \dist_H(t, w) + \dist_H(w, v)\\
    &= \dist_H(u, v) + 2\dist_H(t, w),
  \end{align*}
  where the last equality holds as $w$ lies on a shortest $uv$-path
  $P$ in $G$, which is also a shortest $uv$-path in $H$, since $H$ is
  an induced subgraph of $G$ that contains all vertices of $P$.  For
  the same reason, we get $\dist_H(u, v) = \dist_G(u, v)$, which
  proves the claim.
\end{proof}

Consider the setting as in the above lemma, let $w$ be chosen such
that $\dist_H(t, w)$ is minimal and let $\xi = 2\dist_H(t, w) /
\dist_G(u, v)$.  Then, it holds that $\dist_T(u, v) \le (1 +
\xi)\dist_G(u, v)$.  That is, $T$ has stretch $(1 + \xi)$.  The
following algorithm computes a tree-cover with the same stretch.

Let $G$ be the input graph.  The algorithm operates in phases,
starting with phase~$0$.  For each phase $i$, we define a radius $r_i
= b^i$, for a base $b > 1$.  Then, for $a > 0$, we choose a vertex~$t$
in the current graph and compute the partial shortest-path tree with
root $t$ containing all vertices with distance at most $ (1 + a) r_i$
from $t$.  Afterwards, we delete all vertices with distance at most
$r_i$ to $t$ from the current graph.  This is iterated until all
vertices are deleted.  Afterwards, phase $i$ is done and we restore
the original input graph $G$ before starting phase $i + 1$.  This
process is stopped, once the whole graph is deleted after processing
the first tree in a phase.  The output of the algorithm is the set of
all trees computed during execution.  Since the algorithm
\textbf{pro}duces \textbf{t}ree-covers \textbf{o}f \textbf{n}etworks,
we call it \treecompalgo.

Note that \treecompalgo has several degrees of freedom.  We can choose
the parameters $a > 0$ and $b > 1$, as well as the order in which the
roots of the partial shortest-path trees are selected.  The following
lemma holds independent of the root selection strategy.

\begin{lemma}
  \label{lem:tree-cover-algo-stretch}
  The tree-cover computed by \treecompalgo has stretch $(1 + 2b/a)$
  with additive bound~$2$.
\end{lemma}
\begin{proof}
  Let $\mathcal{C}$ be the tree-cover computed by \treecompalgo, let
  $G = (V, E)$ be the input graph, and let $u \neq v \in V$ be two
  arbitrary vertices.  We have to show that $\mathcal{C}$ contains a
  tree $T$ that includes $u$ and $v$ such that
  $\dist_T(u, v) \le (1 + 2b/a) \dist_G(u, v)$ or
  $\dist_T(u, v) \le \dist_G(u, v) + 2$.

  Let $i$ be minimal such that $\dist_G(u, v) \le a r_i$.  Assume for
  now that \treecompalgo did not stop before phase $i$; we deal with
  the other case later.  As phase $i$ continues until all vertices are
  deleted, at one point a vertex $w$ on a shortest $uv$-path in $G$ is
  deleted.  Let $H$ be the current graph before that happens for the
  first time and let $T$ be the partial shortest-path tree computed in
  $H$ rooted at $t$.  To show that $T$ is the desired tree, we aim to
  apply Lemma~\ref{lem:u-v-covering-ball-tree}.

  First note that $H$ is an induced subgraph of $G$ that contains all
  vertices on a shortest $uv$-path of $G$.  Moreover, $T$ contains $u$
  and $v$ for the following reason.  As $w$ is deleted, we have that
  $\dist_H(t, w) \le r_i$.  Moreover, as $w$ lies on a shortest
  $uv$-path, the distance from $w$ to either $u$ or $v$ cannot exceed
  $\dist_H(u, v) = \dist_G(u, v)$.  Thus, by the triangle inequality
  and the above choice of $i$, we have
  $\dist_H(t, u) \le \dist_H(t, w) + \dist_H(w, u) \le r_i + a r_i =
  (1 + a) r_i$, which implies that $u$ is a vertex of $T$.
  Analogously, $v$ is also contained in $T$.

  With this, we can apply Lemma~\ref{lem:u-v-covering-ball-tree},
  yielding stretch $(1 + \xi)$ for
  $\xi = 2\dist_H(t, w) / \dist_G(u, v)$.  To bound $\xi$, recall that
  we chose $i$ minimal such that $\dist_G(u, v) \le a r_i$.  Thus, if
  $i > 0$, then $\dist_G(u, v) > a r_{i - 1} = \frac{a}{b} r_i$.
  Together with the fact that $\dist_H(t, w) \le r_i$, we obtain
  $\xi \le 2 \frac{b}{a}$, as desired.  In the special case that
  $i = 0$ we have $r_i = 1$ and therefore $\dist_H(t, w) \le 1$.
  Thus, Lemma~\ref{lem:u-v-covering-ball-tree} directly yields
  $\dist_T(u, v) \le \dist_G(u, v) + 2$, which is covered by the
  additive bound~$2$.

  Finally, we assumed above that \treecompalgo did not stop before
  phase $i$ and it remains to consider the case where it stops in
  phase $j < i$.  In this case, let $T$ be the tree we get in phase
  $j$, which includes all vertices of $G$.  Let $t$ be the root of
  $T$.  As all vertices have distance at most $r_j$ from $t$, we get
  $\dist_T(u, v) \le 2r_j$.  Moreover, as $i$ was chosen minimal such
  that $\dist_G(u, v) \le a r_i$, we have $\dist_G(u, v) > a r_j$.
  Together with the previous inequality, this gives a stretch of
  $2/a$, which is smaller than the desired $(1 + 2b/a)$, as $b > 1$.
\end{proof}

\subsection{Performance on Strongly Hyperbolic Unit Disk Graphs}
\label{sec:hyperbolic-analysis}

Throughout the remainder of the paper, we utilize strongly hyperbolic
unit disk graphs to investigate the performance of routing on networks
with underlying hyperbolic geometry.  This is not only interesting
since routing is one of the most fundamental graph problems, but is
also particularly relevant on complex networks like the internet,
which has previously been observed to fit well into the hyperbolic
plane~\cite{bpk-si-10}.

While \treecompalgo computes a $(c, \Additive, k)$-tree-cover with
bounded stretch, the value $k$, i.e., the maximum number of trees that
a vertex is contained in, depends on the structure of the considered
graph.  In the following, we show that $k$ is small on networks with
an underlying hyperbolic geometry.  In our analysis, we consider the
\emph{radially increasing} root selection strategy that selects the
vertices in order of increasing distance to the origin of the
hyperbolic plane, and prove the following theorem.  There, $\diam(G)$
denotes the \emph{diameter} of $G$, i.e., the length of the longest
shortest path in $G.$

\begin{theorem}
  \label{thm:tree-cover-hyperbolic-udg}
  Let $G$ be a strongly hyperbolic unit disk graph with radius
  $R > 0$.  Given the disk representation of $G$, $a > 0$, and
  $b > 1$, the \treecompalgo algorithm with the radially increasing
  root selection strategy computes a $(c, \Additive, k)$-tree-cover
  of~$G$ with
  \begin{equation*}
    c = 1 + 2b/a,~\Additive = 2~\text{, and}~k = \pi e \left( \frac{1 + a}{b - 1} (b^2 \diam(G) - 1) R + 2\left( \log_b(\diam(G)) + 2 \right) \right). \qedhere
  \end{equation*}
\end{theorem}
First note that the correctness of the claimed stretch immediately
follows from \Cref{lem:tree-cover-algo-stretch}.  However, bounding
$k$ is more involved.  To that end, we first compute an upper bound
that holds for a given phase and afterwards sum over all phases.

Consider the roots of the partial shortest-path trees that contain a
vertex $v$ in a given phase, which we refer to as the \emph{roots of
  $v$}.  We partition the hyperbolic disk into radial bands and
compute an upper bound on the number of roots of $v$ in each band, see
Figure~\ref{fig:bands} for an illustration.  We then utilize two key
ingredients.  First, since $v$ is contained in the partial
shortest-path trees of its roots, the length of the path between $v$
and a root is bounded, and so is the angular distance between them.
Consequently, all roots in a band lie in a bounded angular interval
(blue areas in Figure~\ref{fig:bands}).  Secondly, roots cannot be
adjacent as they would otherwise delete each other, which means that
the hyperbolic distance between them has to be sufficiently large.
For roots in the same band, this can only be achieved if their angular
distance is large.  Consequently, each root in a band reserves a
portion of the angular interval (red areas in Figure~\ref{fig:bands})
that no other root can lie in, from which we can derive an upper bound
on the number of roots that lie in the band.

\begin{figure}
  \centering
  \includegraphics{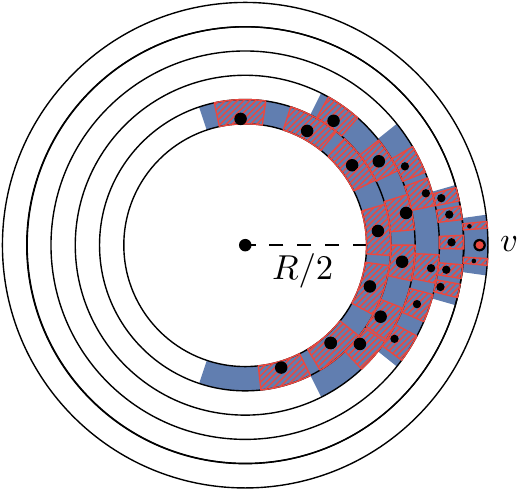}
  \caption{The hyperbolic disk is divided into radial bands.  The
    roots (black vertices) of $v$ (red dot) in a band lie in an
    angular interval $\Phi$ of bounded width (blue).  Each root
    reserves a portion of that interval (red) that no other root can
    lie in.  All vertices with radius at most $R/2$ are removed after
    processing the first root.}
  \label{fig:bands}
\end{figure}

The following lemma bounds the angular distance between a vertex $u$
and another vertex~$u_k$, assuming that there exists a path of length
$k$ between them that consists only of vertices whose radius is not
smaller than the one of $u$.  In particular, this applies to roots
of~$v$: In a given phase, the length of the paths considered in the
partial shortest-path trees is bounded.  Moreover, when the partial
shortest-path tree of a root $\rho$ of $v$ is computed, all vertices
of smaller radii than $\rho$ have been deleted (since roots are
considered in order of increasing radius), meaning the path from
$\rho$ to $v$ cannot contain vertices of smaller radius.

\begin{lemma}
  \label{lem:path-bounded-angle}
  Let $G$ be a strongly hyperbolic unit disk graph with radius $R$ and
  let $u$ be a vertex with $r(u) \ge R/2$.  Further, let
  $P = (u, u_1, \dots, u_k)$ be a path with $r(u) \le r(u_i)$ for all
  $i \in [k]$.  Then,
  $\Delta_{\varphi}(u, u_k) \le k \cdot \pi e^{R/2 - r(u)}$.
\end{lemma}
\begin{proof}
  For convenience, we define $u_0 = u$.  Then, $\Delta_{\varphi}(u,
  u_k)$ can be bounded by
  \begin{align*}
    \Delta_{\varphi}(u, u_k) \le \sum_{i = 1}^{k} \Delta_{\varphi}(u_{i - 1}, u_i).
  \end{align*}
  Note that $u_{i-1}$ and $u_i$ are adjacent and recall that
  $\theta(r(u_{i - 1}), r(u_i))$ denotes the maximum angular distance
  between them, such that this is the case.  Thus, 
  \begin{align*}
    \Delta_{\varphi}(u, u_k) \le \sum_{i = 1}^{k} \theta(r(u_{i - 1}), r(u_i)).
  \end{align*}
  Since $R/2 \le r(u) \le r(u_i)$ for all $i \in [k]$ is a
  precondition of this lemma, we have $r(u_{i - 1}) + r(u_i) \ge R$
  for all $i \in [k]$.  Consequently, we can apply
  \Cref{lem:tight-max-angle-adjacent} to bound
  $\theta(r(u_{i - 1}), r(u_i))$, which yields
  \begin{align*}
    \Delta_{\varphi}(u, u_k) &\le \sum_{i = 1}^{k} \pi e^{(R - r(u_{i - 1}) - r(u_i))/2} \le \sum_{i = 1}^{k} \pi e^{(R - r(u) - r(u))/2} = k \cdot \pi e^{R/2 - r(u)},
  \end{align*}
  where the second inequality is valid since $r(u) \le r(u_i)$ for all
  $i \in [k]$.
\end{proof}

The second key ingredient is a lower bound on the minimum angular
distance between two non-adjacent vertices in a radial band of fixed
width in the hyperbolic disk.  We note that in order to obtain a bound
that is easy to work with, we aim to utilize
\Cref{lem:max-angle-adjacent}.  However, this requires that $R$ is not
too small.  For now, we assume that this requirement is met and
afterwards resolve the constraint in the analysis of the algorithm.

\begin{lemma}
  \label{lem:constant-radial-increase-large-angle}
  Let $G$ be a strongly hyperbolic unit disk graph with radius
  $R \ge 1$ and let $r \ge R/2$ be a radius.  Further, let $u, v$ be
  non-adjacent vertices with $r(u), r(v) \in [r, r + \tau]$ for
  $\tau \in [0, R - 1]$.  Then,
  $\Delta_{\varphi}(u, v) \ge e^{R/2 - (r + \tau)}$.
\end{lemma}
\begin{proof}
  Recall that $\theta(r(u), r(v))$ denotes the maximum angular
  distance such that~$u$ and~$v$ are adjacent.  Since the two vertices
  are not adjacent in our case, we can derive that
  $\Delta_{\varphi}(u, v) > \theta(r(u), r(v))$.  We now aim to
  apply~\Cref{lem:max-angle-adjacent} in order to obtain a lower bound
  on $\theta(r(u), r(v))$.  To this end, we first validate that its
  preconditions are met.  Since $r(u), r(v) \ge r \ge R/2$, we have
  $r(u) + r(v) \ge R$.  Moreover, by assumption we know that
  $r(u), r(v) \in [r, r + \tau]$ for $\tau \in [0, R - 1]$, which
  implies that $|r(u) - r(v)| \le \tau \le R - 1$.  Consequently, we
  can apply~\Cref{lem:max-angle-adjacent} to conclude that
  \begin{align*}
    \theta(r(u), r(v)) &\ge e^{(R - r(u) - r(v))/2} \\
                       &\ge e^{(R - 2r - 2\tau)/2} \\
                       &= e^{R/2 - (r + \tau)},
  \end{align*}
  where the second inequality is valid, since by assumption
  $r(u), r(v) \le r + \tau$.
\end{proof}

We can now combine the two key ingredients to compute an upper bound
on the number of the roots of $v$ in a given phase $i$, which we
denote by $\bm{\rho}_i(v)$. 

\begin{lemma}
  \label{lem:bounded-labels-in-a-phase}
  Let $G$ be a strongly hyperbolic unit disk graph with radius
  $R > 0$.  Let the disk representation of $G$, $a > 0$, and $b > 1$
  be given and consider phase $i$ of the \treecompalgo algorithm.
  Then, for every vertex $v$ it holds that
  $|\bm{\rho}_i(v)| \le \pi e (R (1 + a)b^{i} + 2)$.
\end{lemma}
\begin{proof}
  In the following, we aim to
  utilize~\Cref{lem:path-bounded-angle,lem:constant-radial-increase-large-angle},
  both of which require that the considered vertices have a radius of
  at least $R/2$ and one additionally assumes that $R \ge 1$
  Therefore, we first argue about the case where these conditions are
  not met.  First note that after the first root in a phase is
  processed, all vertices with radius at most $R/2$ are removed since
  (if they exist in the first place) they form a clique.
  Additionally, when $R < 1$, the whole graph can be covered by few
  cliques.  More precisely,
  by~\Cref{lem:strongly-unit-disk-clique-cover} a strongly hyperbolic
  unit disk graph with radius $R$ can be covered by
  $\max\{2\pi \sqrt{2}, 2 \pi e^{R/2}\}$ cliques.  In particular, for
  $R < 1$, this yields a bound of $2 \pi \sqrt{e}$.  Since processing
  each root removes at least one such clique from the graph, the
  number of roots in the phase is bounded by the number of cliques.
  It follows, considering the first clique in $D_{R/2}$ and the
  remaining cliques when $R < 1$, that we can bound the roots of $v$
  in phase $i$ as $|\bm{\rho}_i(v)| \le 1 + 2\pi\sqrt{e} \le 2\pi e$,
  which we account for with the $+2$ in the lemma statement.

  For the remaining roots of $v$ we can now assume that $R \ge 1$ and
  that all vertices have radius at least $R/2$.  Furthermore, it
  suffices to show that there are at most $\pi e R (1 + a)b^i$ such
  roots.  We cover the remainder of the disk with $R/2$ bands of
  radial width $1$, where the $j$th band (for
  $j \in \{0, \dots, R/2 - 1\}$) contains all points with radius in
  $[R/2 + j, R/2 + j + 1]$, see \Cref{fig:bands}.  The claim then
  follows if we can bound the number of roots of $v$ in a single band
  by $2 \pi e (1 + a)b^{i}$.

  \begin{figure}
    \centering
    \includegraphics[scale=1.0]{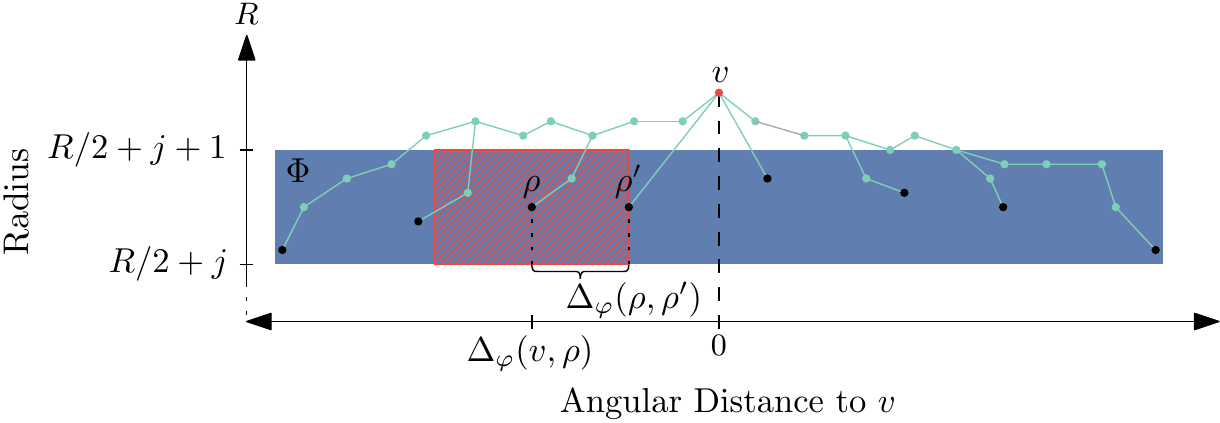}
    \caption{A vertex $v$ (red dot) and the roots (black vertices)
      that are contained in the $j$th band and are connected to $v$
      (via the green paths).  All roots lie in the angular interval
      $\Phi$ (blue region).  Other than $\rho$, no root can lie in the
      red region.}
    \label{fig:root-sequence}
  \end{figure}

  Let $\bm{\rho}_{i, j}(v)$ denote the roots of $v$ that lie in the
  $j$th band (see Figure~\ref{fig:root-sequence}).  We first bound the
  angular distance between $v$ and a root in $\bm{\rho}_{i, j}(v)$,
  and with that the width of the angular interval $\Phi$ (blue region
  in Figure~\ref{fig:root-sequence}) that contains all of them.
  Afterwards, we show that each root reserves a portion of $\Phi$ that
  no other root can be in.  An upper bound on $|\bm{\rho}_{i, j}(v)|$
  is then obtained by the quotient of the widths of the two intervals.

  Consider a root $\rho \in \bm{\rho}_{i, j}(v)$.  Since the roots are
  processed in order of increasing radius, all vertices of radius at
  most $r(\rho)$ have been removed before.  Consequently, the path
  from~$\rho$ to $v$ in the partial shortest-path tree rooted at
  $\rho$ consists only of vertices with radius at least~$r(\rho)$.
  Moreover, in phase $i$ the depth of this tree is $(1 + a)b^i$, which
  means that the path between $\rho$ and $v$ is at most this long.
  Therefore, we can apply Lemma~\ref{lem:path-bounded-angle} to
  conclude that the maximum angular distance between $v$ and a root
  $\rho$ is at most
  \begin{align*}
    \max_{\rho \in \bm{\rho}_{i, j}(v)} \Delta_{\varphi}(v, \rho) \le \max_{\rho \in \bm{\rho}_{i, j}(v)} (1 + a)b^i \cdot \pi e^{R/2 - r(\rho)} \le (1 + a)b^i \cdot \pi e^{-j},
  \end{align*}
  where the last inequality stems from the fact that $r(\rho) \ge R/2
  + j$ holds for all $\rho \in \bm{\rho}_{i, j}(v)$.  Moreover, since
  roots cannot be adjacent (as they would otherwise delete each other)
  and all roots in $\bm{\rho}_{i, j}(v)$ have their radii in $[R/2 +
  j, R/2 + j + 1]$, we can apply
  Lemma~\ref{lem:constant-radial-increase-large-angle} to conclude
  that the minimum angular distance between two roots $\rho, \rho' \in
  \bm{\rho}_{i, j}(v)$ is at least
  \begin{align*}
    \min_{\rho \neq \rho' \in \bm{\rho}_{i, j}(v)}\Delta_{\varphi}(\rho, \rho') \ge e^{R/2 - (R/2 + j + 1)} = e^{-(j + 1)}.
  \end{align*}
  Note that the angular interval $\Phi$ extends to both angular
  directions from $v$. 
  Therefore, the number of roots in $\bm{\rho}_{i, j}(v)$ can be
  bounded by
  \begin{align*}
    |\bm{\rho}_{i, j}(v)| &\le 2 \cdot \frac{\max_{\rho \in \bm{\rho}_{i, j}(v)} \Delta_{\varphi}(v, \rho)}{\min_{\rho \neq \rho' \in \bm{\rho}_{i, j}(v)} \Delta_{\varphi}(\rho, \rho')} \\
                          &\le \frac{2 (1 + a)b^{i} \cdot \pi e^{-j}}{e^{-(j + 1)}} \\
                          &= 2 \pi e(1 + a)b^i.\qedhere
  \end{align*}
\end{proof}

With this we are now ready to bound the number $k$ of trees that a
vertex is contained in, for tree-covers produced by the \treecompalgo
algorithm.

\begin{proof}[Proof of Theorem~\ref{thm:tree-cover-hyperbolic-udg}]
  First note that the values for $c$ and $\Additive$ hold for any
  graph due to~\Cref{lem:tree-cover-algo-stretch}.  It remains to show
  that the stated bound on $k$ is valid.  To that end, we make use
  of~\Cref{lem:bounded-labels-in-a-phase}, which states that $v$ is
  contained in at most $\pi e (R (1 + a)b^{i} + 2)$ trees in phase
  $i$, and sum over all phases.  Since the radius of the shortest-path
  trees that are removed from the graph in two consecutive phases
  increases by a factor of $b$ and the algorithm terminates when the
  first tree in a phase deletes the whole graph, there are at most
  $\lceil{\log_b(\diam(G))}\rceil$ phases.  Thus,
  \begin{align*}
    k &= \sum_{i = 0}^{\log_b(\diam(G)) + 1} \pi e (R (1 + a)b^i + 2) \\
      &= \pi e \left( R (1 + a) \left( \sum_{i = 0}^{\log_b(\diam(G)) + 1} b^i \right) + 2\left( \log_b(\diam(G)) + 2 \right) \right).
  \end{align*}
  Note that the remaining sum is a partial sum of a geometric series
  with $b > 1$, which can be computed as
  $\sum_{i = 0}^{x}b^i = (b^{x + 1} - 1)/(b - 1)$.  We obtain
  \begin{align*}
    k &= \pi e \left( R (1 + a) \frac{b^{\log_b(\diam(G)) + 2} - 1}{b - 1} + 2\left( \log_b(\diam(G)) + 2 \right) \right) \\
      &= \pi e \left( \frac{1 + a}{b - 1} (b^2 \diam(G) - 1) R + 2 \left( \log_b(\diam(G)) + 2 \right) \right). \qedhere
  \end{align*}
\end{proof}

Since hyperbolic random graphs are a special case of strongly
hyperbolic unit disk graphs, where vertices are distributed in a disk
of radius $R = \mathcal{O}(\log n)$ and since these graphs have a
diameter of $\mathcal{O}(\log n)$ asymptotically almost
surely~\cite{ms-k-19}, we obtain the following corollary.

\begin{corollary}
  Let $G$ be a hyperbolic random graph.  Given the disk representation
  of $G$, $a > 0$, and $b > 1$, the \treecompalgo algorithm with the
  radially increasing root selection strategy computes a
  $(c, \Additive, k)$-tree-cover of $G$ with $c = 1 + 2b/a$,
  $\Additive = 2$, and, asymptotically almost surely
  \begin{align*}
    k = \mathcal{O}\left( \frac{(1 + a) b^2}{b - 1} \cdot \log^2 n \right).
  \end{align*}
\end{corollary}

Together with Theorem~\ref{thm:cover-based-routing}, it follows that
greedy routing on hyperbolic random graphs can be implemented such
that the resulting scheme is starvation-free and has stretch
$1 + 2b / a$ with additive bound $2$.  Moreover, by setting
$a = b = 2$ we obtain a multiplicative stretch of $3$, and can derive
that the scheme, asymptotically almost surely, stores
$\mathcal{O}(\log^4 n)$ bits at each vertex and takes
$\mathcal{O}(\log^2 n)$ time per query, which improves upon the
performance lower bound for general graphs.

\section{Experiments}
\label{apx:experiments}

\begin{figure}[t]
  \centering
  \input{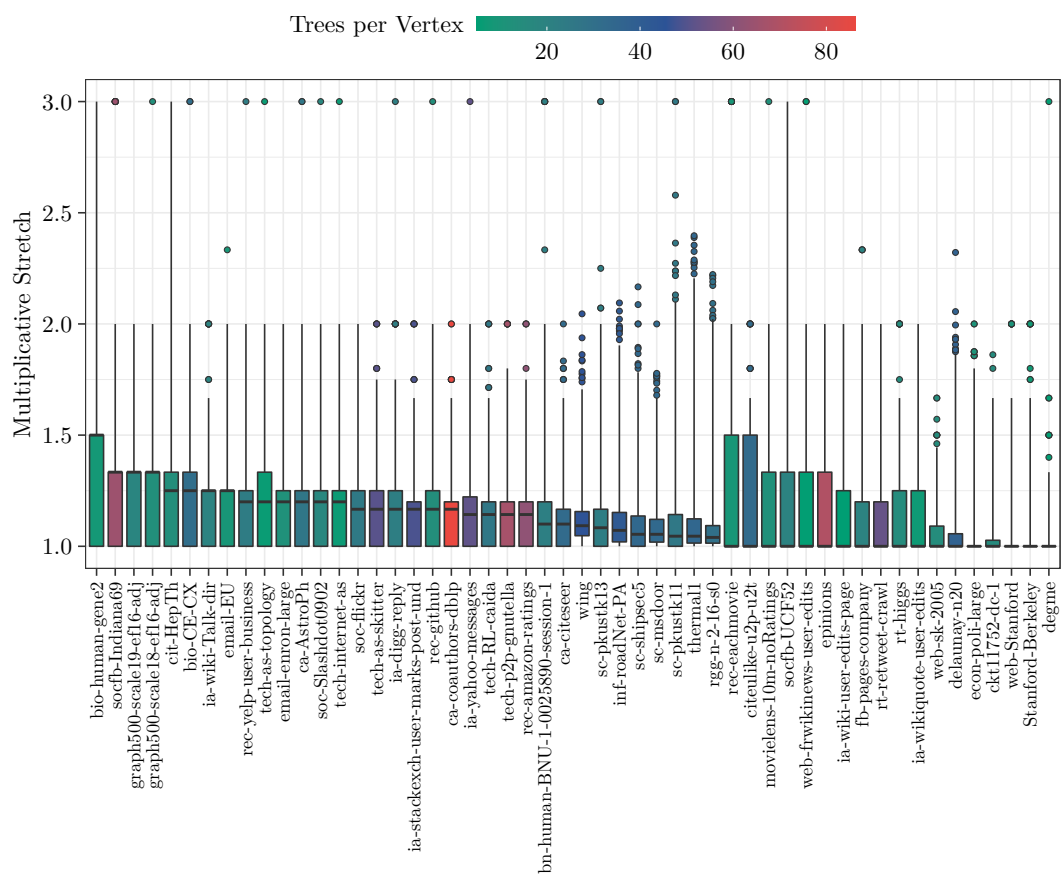}
  \caption{Multiplicative stretch when routing with a tree-cover
    obtained using \treecompalgo with $a = b = 2$.  Colors show the
    number of trees that an average vertex is contained in.  Boxes
    denote the interquartile range extending to the $25$th and $75$th
    percentile with horizontal bars showing the median.  Whiskers
    extend to $0.1\%$ and $99.9\%$, while dots show values beyond
    that.}
  \label{fig:experiments}
\end{figure}

We designed a routing scheme that utilizes hierarchical structures and
showed that it has small stretch, space requirements, and query times
on strongly hyperbolic unit disk graphs.  To evaluate how well our
results translate to real-world networks, we performed experiments on
$50$ graphs from the Network Data Repository~\cite{ra-ndrigav-15},
with sizes ranging from \SI{14}{k} to over \SI{2.3}{M} vertices.
Since we do not have unit disk representations for these, we used the
degrees of the vertices as a proxy for their place in the hierarchical
structure.
That is, the root selection strategy processed the vertices by
decreasing degree.  For each graph, we computed a tree-cover using the
\treecompalgo algorithm with parameters $a = b = 2$, and sampled
\SI{10}{k} vertex pairs for which the path obtained by our routing
scheme was compared to a shortest path between them.
Figure~\ref{fig:experiments} shows boxplots aggregating our
observations.

As expected, the maximum observed stretch is $3$.  However, this
stretch occurred only rarely.  In all networks most of the sampled
routes had a stretch of at most $1.5$ and in $16$ of the $50$ graphs
the median stretch was $1$.  At the same time, the number of trees
that a vertex was contained in on average remained small.  In $42$ of
the $50$ networks this number was less than $50$, even in networks
with over \SI{2.3}{M} vertices.



\bibliography{forest_routing}

\end{document}